\documentclass[sigplan,10pt,nonacm, twocolumn]{acmart}

\usepackage{amsmath}
\usepackage{amsthm}
\usepackage{proof}
\usepackage{graphicx}
\usepackage{xspace}
\usepackage{color} 
\usepackage[all]{xy}
\usepackage{tikz-cd}
\usepackage{hyperref}
\usepackage{mathtools}
\usepackage{macros}
\usepackage{MnSymbol}
\usepackage{todonotes}
\usepackage[shortlabels]{enumitem}

\usepackage{libertinus}

\newcommand{\upgrade}[1]{\ !{#1}}
\newcommand{\upgradelong}[1]{\ !(#1)}

\title{Folding interpretations}
\author{Miko{\l}aj Boja\'nczyk (University of Warsaw)}
\titlenote{This is the author's version of a LICS 2023 paper.}
\begin{document}

\begin{abstract}
    We study the polyregular string-to-string functions, which are certain functions of polynomial output size that can be described using automata and logic.
    We describe a system of combinators that generates exactly these functions. Unlike previous systems, the present system includes an iteration mechanism, namely fold. Although unrestricted fold can define all primitive recursive functions, we identify a type system (inspired by linear logic) that restricts fold so that it defines exactly the polyregular functions. We also present related systems, for quantifier-free functions as well as for linear regular functions on both strings and trees. 
\end{abstract}

\maketitle

\pagestyle{plain}

\newcommand{\primefunction}[4]{
    #1 &  #2 #3 & \text{#4}
}

\newcommand{\primeqf}{
    \primefunction{\Gamma \times \Sigma} \leftrightarrow {\Sigma \times \Gamma}{commutativity of $\times$}\\
    \primefunction{\Gamma + \Sigma} \leftrightarrow  {\Sigma + \Gamma} {commutativity of $+$}\\
\primefunction{\Gamma \times (\Sigma \times \Delta)}  \leftrightarrow  {(\Gamma \times \Sigma ) \times \Delta} {associativity of $\times$} \\
\primefunction{\Gamma + (\Sigma + \Delta)}  \leftrightarrow  {(\Gamma + \Sigma ) + \Delta} {associativity of $+$} \\
\primefunction{\Gamma \times (\Sigma + \Delta)} \leftrightarrow  {(\Gamma \times \Sigma) + (\Gamma \times \Delta)}{distributivity}\\
\primefunction{\Gamma_1 \times \Gamma_2}  \to {\Gamma_i}  {projections}\\
\primefunction{\Gamma_i} \to {\Gamma_1 + \Gamma_2} {co-projections}\\
\primefunction{\Gamma + \Gamma}  \to  \Gamma {co-diagonal} \\
    }

    \newcommand{\primepr}{
        \primefunction{\Gamma \times \Sigma} \leftrightarrow {\Sigma \times \Gamma}{commutativity of $\times$}\\
        \primefunction{\Gamma + \Sigma} \leftrightarrow  {\Sigma + \Gamma} {commutativity of $+$}\\
    \primefunction{\Gamma \times (\Sigma \times \Delta)}  \leftrightarrow  {(\Gamma \times \Sigma ) \times \Delta} {associativity of $\times$} \\
    \primefunction{\Gamma + (\Sigma + \Delta)}  \leftrightarrow  {(\Gamma + \Sigma ) + \Delta} {associativity of $+$} \\
    \primefunction{\Gamma \times (\Sigma + \Delta)} \leftrightarrow  {(\Gamma \times \Sigma) + (\Gamma \times \Delta)}{distributivity}\\
    \primefunction{\Gamma_1 \times \Gamma_2}  \to {\Gamma_i}  {projection, for $i = 1,2$}\\
    \primefunction{\Gamma_i} \to {\Gamma_1 + \Gamma_2} {co-projection, for $i = 1,2$}\\
    \primefunction{\Gamma + \Gamma}  \to  \Gamma {co-diagonal} \\
    \primefunction{\Gamma}{\to}{1}{constant}\\
    \primefunction {\Gamma^{+} \times \Gamma^+}  \to {\Gamma^+} {concat}\\
    \primefunction {\Gamma}  \to  {\Gamma^+}  {list unit}
    }

\newcommand{\combqf}
{
    \frac{\Gamma_1 \to \Sigma_1 \quad \Gamma_2 \to \Sigma_2}{\Gamma_1 \times \Gamma_2 \to \Sigma_1 \times \Sigma_2}   & & \text{functoriality of $\times$} \\[1.5ex]
        \frac{\Gamma_1 \to \Sigma_1 \quad \Gamma_2 \to \Sigma_2}{\Gamma_1 + \Gamma_2 \to \Sigma_1 + \Sigma_2} 
           & & \text{functoriality of $+$} \\[1.5ex] 
             \qquad \frac{\Gamma \to \Sigma} { \Gamma^*  \to  \Sigma^*} & & \text{functoriality of $*$} \\[1.5ex] 
}

\newcommand{\primepoly}{
    \primefunction{\ \upgradelong{\Gamma + \Sigma}}{\leftrightarrow}{\upgrade \Gamma + \upgrade  \Sigma }{{upgrade commutes with $+$}}\\
    \primefunction{\ \upgradelong{\Gamma \times \Sigma}}{\leftrightarrow}{\upgrade \Gamma \times \upgrade  \Sigma }{{upgrade commutes with $\times$}}\\
    \primefunction{\ (\upgrade \Gamma)^+}{\leftrightarrow}{\upgradelong{\Gamma^+}}{{upgrade commutes with lists}}\\
    \primefunction{\ \upgrade{\Gamma}}{\rightarrow}{\upgrade \Gamma \times \Gamma  }{{copying}}
}
\newcommand{\combpoly}{
    \frac{\Gamma \to \Sigma} {\upgrade \Gamma \to \upgrade \Sigma}
        & & \text{{functoriality of upgrade}} \\[1.5ex]
         \frac{ \Sigma + \Gamma \times  \Sigma \to \Gamma }{ \upgradelong{\Sigma^+} \to \Gamma} & & \text{{fold}}
}

\newcommand{\foldcombinator}
{
    \frac{!^k 1 \to \Gamma \qquad \Gamma \times  \Sigma \to \Gamma}{!^k\Sigma^* \to \Gamma} & & \text{safe fold}
}

\newcommand{\foldcombinatorinit}
{
    \frac{1 + \Gamma \times  \Sigma \to \Gamma}{!\Sigma^* \to \Gamma} & & \text{initialized fold}
}

\newcommand{\foldcombinatornobang}
{
    \frac{\Gamma \times  \Sigma \to \Gamma}{\Gamma \times  \Sigma^* \to \Gamma} & & \text{fold}
}

\newcommand{\combinator}[3]{
    \frac{#1}{#2} & \qquad \text{#3}
}

\newcommand{\generalfold}
{
    \combinator{ 1 \to \Gamma \quad \Gamma \times \Sigma \to \Gamma}{ \Sigma^* \to \Gamma}{fold}
}

\section{Introduction}
\label{sec:introduction}
This paper is about transducers that compute string-to-string functions. (We also have some results on trees, but trees will be discussed only at the end of the paper. ) We are interested in two classes of functions: the linear regular functions\footnote{These are usually called the regular functions in the literature, but we add the word ``linear'' to distinguish them from the polyregular functions.}, which have linear output size, and the polyregular functions, which have polynomial output size. Both classes can be described by many equivalent models, and have robust closure properties. 

Let us begin with the more established class of linear regular functions.  Two typical example functions from this class are:
\begin{align*}
\myunderbrace{[1,2,3] \mapsto [1,2,3,1,2,3] }{duplicate} 
\qquad 
\myunderbrace{[1,2,3] \mapsto [3,2,1] }{reverse} .
\end{align*}
The linear regular functions can be described by many equivalent models, including: deterministic two-way automata with output~\cite[Note 4]{shepherdson1959reduction}, \mso transductions~\cite[Section 4]{engelfrietMSODefinableString2001}, 
streaming string transducers~\cite[Section 3]{alurExpressivenessStreamingString2010}, an extension of regular expressions~\cite[Section 2]{alur2014regular}, and a calculus based on combinators~\cite[Theorem 6.1]{bojanczykRegularFirstOrderList2018}. 
 The many equivalent models, as well as the robustness and good decidability properties of the underlying class, are comparable to similar properties for the regular languages, which also have many equivalent descriptions, including automata, logic and regular expressions. For this reason, the linear regular functions have been intensively studied in the last decade.

The second class is the polyregular functions, which  extended the linear regular functions by allowing polynomial growth, including functions such as the squaring operation
 \begin{align*}
 [1,2,3] \mapsto [1,2,3,1,2,3,1,2,3].
 \end{align*}
Similarly to the linear regular functions, the polyregular functions can also be described by multiple models, including: string-to-string pebble transducers, which are introduced in~\cite[Section 1]{engelfriet2002two} based on~\cite[Definition 1.5]{DBLP:journals/tcs/GlobermanH96} and~\cite[Section 3.1]{DBLP:journals/jcss/MiloSV03}, as well as an imperative programming language~\cite[Section 3]{bojanczykPolyregularFunctions2018}, a functional programming language~\cite[Section 4]{bojanczykPolyregularFunctions2018}, and a polynomial extension of \mso transductions~\cite[Definition 2]{msoInterpretations}. For a survey of the polyregular functions, see~\cite{polyregular-survey}.

\paragraph*{Combinators.} This paper studies the linear regular and polyregular functions by using systems based on prime functions and combinators. This approach dates back to the Krohn-Rhodes Theorem~\citep[p.~454]{Krohn1965}, and was first applied to linear regular functions in~\cite{bojanczykRegularFirstOrderList2018}, by describing them in terms of certain prime functions, such as 
\begin{align*}
\primefunction{1 + \Sigma \times \Sigma^*}{\to}{\Sigma^*}{list constructor,}
 \end{align*}
and  combinators such as
\begin{align*}
  \combinator{\Sigma \to \Gamma \quad \Gamma \to \Delta} { \Sigma \to \Delta}{function composition.}
\end{align*}
This system is further extended in~\cite[p.~64]{bojanczykPolyregularFunctions2018} to cover the polyregular functions, by adding extra prime functions of non-linear output size, such as the squaring operation.

The systems in~\cite{bojanczykRegularFirstOrderList2018,bojanczykPolyregularFunctions2018} have no constructions for iteration; because of this design decision, the hard part is proving completeness: every function of interest can be derived in the system. One reason for avoiding iteration is to have a minimal system.  Another reason is that iteration constructions are powerful, and as we find out in this paper, it is hard to add them while retaining soundness (only functions of interest can be derived).

\paragraph*{The fold combinator.}
In this paper, we take the opposite approach, by studying an iteration construction, namely the fold combinator. This combinator can be written as a rule 
\begin{align*}
\generalfold.
\end{align*} 
The assumption of this rule can be seen as a deterministic automaton with input alphabet $\Sigma$ and state space $\Gamma$, given by its initial state and transition function. In the conclusion of the rule, we have the function that maps an input string to the last state of the run of the automaton.
The input alphabet and the state space need not be finite, e.g.~the state space $\Gamma$ could be the set $1^*$ which represents the natural numbers. 

Folding is a fundamental construction in functional programming languages. For example, the fold combinator arises canonically from the inductive definition of the list type~\cite[Section 3]{hutton1999tutorial}. Unfortunately, there is a price to pay for the power and elegance of the fold combinator: one can use it to derive all primitive recursive functions~\cite[Section 4.1]{hutton1999tutorial}. Therefore, without any further restrictions, the fold combinator falls outside the scope of automata techniques, or any other techniques that can be used to decide semantic properties of programs, such as the halting problem.

This paper is devoted to identifying restrictions on the fold combinator that tame its expressive power. These restrictions are presented as a typing system, which ensures that applications of fold will stay in the class of polyregular functions. In particular, the resulting class of functions shares the decidability properties of the polyregular functions, e.g.~one can decide if a function produces a nonempty output for at least one input.

There are two main contributions in the paper. 

\paragraph*{Quantifier-free interpretations.}
The first contribution is to identify the quantifier-free interpretations as an important class of functions in the context of fold. These are functions on structures in which the universe of the output is a subset of the universe of the input (in particular, the output size is linear), and all relations in the output structure are defined using quantifier-free formulas. 

In Theorem~\ref{thm:fold-quantifier-free} we show that applying the fold combinator to a quantifier-free interpretation yields a function that, although not necessarily quantifier-free, is at least linear regular. This result subsumes several existing results, in particular those about \mso definability of streaming transducers~\cite{alur2014regular,alurStreamingTreeTransducers2017}. 
Although quantifier-free interpretations are rather weak, they can describe most natural transformations that are used as primes in the calculi from~\cite{bojanczykRegularFirstOrderList2018,bojanczykPolyregularFunctions2018}; the remaining primes can then be derived using fold.

Having identified the importance of quantifier-free functions, in Theorem~\ref{thm:qf-system}, we present a system of prime functions and combinators that derives exactly the quantifier-free functions. The completeness proof of the system is the longest proof in the paper.  The quantifier-free system does not allow fold; fold is used in the next part of the paper, about polyregular functions.

\paragraph*{Safe fold.}
The second main contribution is a type system that tames the power of fold. This system uses a type constructor $!$ and bears certain similarities to the parsimonius calculus of Mazza~\cite[Section 2.2]{mazza2015simple}. The latter is part of a field called \emph{implicit computational complexity}, which seeks to describe complexity classes using type systems. An influential example of this kind is a system of Bellantoni and Cook~\cite{bellantoni1992new}, which characterizes polynomial time. The present paper can be seen as part of implicit computational complexity, which targets regular languages instead of Turing complete models, such as logarithmic space or polymomial time. For a more detailed discussion of the connections between regular languages and $\lambda$-calculus, including a pioneering applicaton of linear types, see~\cite{implicit1,NguyenNP21}. 

The usual application of $!$ is to restrict duplication, and this paper is no exception, as  in the following example:
\begin{align*}
\myunderbrace{x \mapsto (x,x)}{not allowed} \qquad
\myunderbrace{!x \mapsto (!x,x)}{allowed}.
\end{align*}
However, apart from restricting duplication, $!$ is also used in this paper to restrict another, more mysterious, resource, namely quantifiers. The idea is that our system uses $!$ to describe functions that are not necessarily quantifier-free, but are similar enough to quantifier-free functions so that the fold combinator can be applied to them. 

 The second main contribution of this paper is Theorem~\ref{thm:main}, which characterizes the polyregular functions using certain prime functions and combinators, in which the types involve $!$ and one of the combinators is fold. In Theorem~\ref{thm:linear} we also show that if we further restrict duplication 
\begin{align*}
 \myunderbrace{!x \mapsto (!x,x)}{not allowed} \qquad
 \myunderbrace{!x \mapsto (x,x)}{allowed},
 \end{align*}
then the resulting system derives exactly the linear functions. Finally, we also show that the results about the linear case can be extended from strings to trees without much difficulty. 

\paragraph*{Acknowledgement.} I would like to thank L{\^{e}} Th{\`{a}}nh Dung Nguy{\^{e}}n and the anonymous reviewers for many helpful comments. This work was financially supported by the Leverhulme Trust, and the Polish National Agency for Academic Exchange.

\section{Interpretations}
In this section, we describe the polyregular functions. 
Among several equivalent definitions of the polyregular functions, our point of departure in this paper will be a definition that uses \mso interpretations~\cite[Section 2]{msoInterpretations}.

\subsection{Definition of \mso interpretations}
 We assume that the reader is familiar with the basics of monadic second-order logic \mso, see~\cite{ebbinghausFlumFinite} for an introduction. We only describe the notation that we use.
A \emph{vocabulary} consists of a finite set of relation names, each one with an associated arity in $\set{0,1,\ldots}$. Note that we allow nullary relations, i.e.~relations of arity zero; such a relation takes no arguments and is ``true'' or ``false'' in each structure. A \emph{structure} over such a vocabulary consists of a finite set, possibly empty, called the \emph{universe} of the structure, and an interpretation of the vocabulary, which associates to each relation name in the vocabulary a relation over the universe of matching arity. The syntax and semantics of  \mso is defined in the usual way.
Whenever we speak of a \emph{class of structures}, all structures in the class must be over the same vocabulary, and the class must be closed under isomorphism. The structures considered in this paper will be used to describe finite strings and similar objects, such as pairs of strings, or strings of pairs of strings.


\paragraph*{Intuitive description.}
We begin with an intuitive description of string-to-string \mso intepretations. Following the classical B\"uchi-Elgot-Trakhtenbrot correspondence of automata and \mso logic, we view strings as structures.

\begin{definition}\label{def:string-as-structure}
 A string in $\Sigma^*$ is viewed as a structure whose universe is the string positions, equipped with the relations
 \begin{align*}
 \myunderbrace{x \le y}{order on positions}
 \hspace{2cm}
 \myunderbrace{a(x)}{$x$ has label $a \in \Sigma$}.
 \end{align*}
\end{definition}
A string-to-string \mso interpretation transforms strings using the above representation, such that the positions of the output string are represented by $k$-tuples of positions in the input string, for some $k \in \set{0,1,\dots}$. The order\footnote{For reasons described in~\cite[Theorem 4]{msoInterpretations}, the string positions are equipped with a linear order $x \ge y$ instead of successor $x = y+1$. } on output positions is defined by a formula 
\begin{align*}
\varphi(\myunderbrace{x_1,\ldots,x_k}{first output\\ position}, \myunderbrace{y_1,\ldots,y_k}{second output\\ position})
\end{align*}
with $2k$ free variables, while the labels of the output positions are defined by formulas with $k$ free variables, one for each letter in the output alphabet. Finally, not all $k$-tuples of input positions need to participate in the output string; there is a formula with $k$ free variables, called the \emph{universe} formula, which selects those that do. All of these formulas need to be consistent -- every $k$-tuple of positions in the input string that satisfies the universe formula must satisfy exactly one of the label formulas, and these $k$-tuples need to be linearly ordered by the order formula. Consistency is decidable, since it boils down to checking if some \mso formula is true in all strings, which in turn boils down to checking if automaton is nonempty by the equivalence of \mso and regular languages.

\paragraph*{Formal definition.} We now give a formal definition of \mso interpretations. The formal definition generalizes the above intuitive description in two ways of minor importance. First, the definition is presented not just for strings, but for general classes of structures; we intend to apply it to mild generalizations of strings, such as pairs of strings or strings of strings. Second, instead of the universe being $k$-tuples of some fixed dimension, it is created using a \emph{polynomial functor}, which is an operation on sets of the form 
\begin{align}\label{eq:polynomial-functor}
 F(A) = A^{k_1} + \cdots + A^{k_n}.
 \end{align} 
Typical polynomial functors include the identity functor $A$, or the functor $A^2+ A^2$ that produces two copies of the square of the input set. 
We use the following terminology for polynomial functors: each $A^{k_i}$ is called a \emph{component} of the polynomial functor, and $k_i \in \set{0,1,\ldots}$ is called the \emph{dimension} of this component. 
This extra generality of polynomial functors\footnote{One can reduce the polynomial functor in an \mso interpretation to a single component $A^k$, at the cost of increasing the dimension $k$. This works for input structures with at least two elements. For this reason, \cite{msoInterpretations} uses interpretations with just one component. } makes the definition more robust, it will be useful in a more refined analysis of \mso interpretations that will appear in Section~\ref{sec:soundness}. In case of linear functors (where all components have dimension at most one), the components correspond to the \emph{copies} in an \mso transduction~\cite[p.~230]{engelfrietMSODefinableString2001}. 

In an \mso interpretation, the polynomial functor is used to define the universe of the output structure; if $A$ is an input structure then elements of $F(A)$ are called  \emph{output candidates}. A subset of the output candidates will be the universe of the output structure. This subset is defined using an \emph{\mso query of type $F$}, which is a family of \mso formulas, with one formula for each component in the functor, such that number of free variables in each formula is the dimension of the corresponding component. Here are some examples:
\begin{align*}
\myunderbrace{A^0 = 1}{a query of this type \\ is a formula without \\ free variables} 
\hspace{2cm}
\myunderbrace{A^4}{a query of this type \\ is a formula with \\ four free variables} 
\hspace{2cm}
\myunderbrace{A^2 + A^2}{a query of this type \\ is two formulas with \\ two free variables each}
\end{align*}
The relations in the output structure are also defined using \mso queries, with a relation of arity $m$ defined using a query of type 
 \begin{align*}
 F^m(A) \eqdef \myunderbrace{F(A) \times \cdots \times F(A)}{$m$ times}
 \end{align*}
The above type is also a polynomial functor,  since polynomial functors are closed under taking products, e.g.~the product of $A^2$ and $A+1$ is $A^3+ A^2$.
The discussion above is summarized in the following definition.

\begin{definition}[\mso interpretation]
 \label{def:mso-interpretation} A function $
 f : \Sigma \to \Gamma$
 between two classes of structures is called an \mso interpretation if:
 \begin{enumerate}
 \item \textbf{Universe.} There is a polynomial functor $F$ and a \mso query of type $F$ such that for every input structure $A \in \Sigma$, the universe of the output structure is the subset of the output candidates $F(A)$ defined by this query; and 
 \item \textbf{Relations.} For every relation name $R$ in the vocabulary of the output class, of arity $m$, there is an \mso query of type $F^m$, which defines the  interpretation of $R$ in every output structure.
 \end{enumerate}
 \end{definition}

A \emph{string-to-string \mso interpretation} is the special case of the above definition where the input type is $\Sigma^*$ for some finite alphabet $\Sigma$, and the output type is $\Gamma^*$ for some finite alphabet $\Gamma$. 

\begin{myexample}
 Consider the squaring operation on strings
 \begin{align*}
 [1,2,3] \mapsto [1,2,3,1,2,3,1,2,3].
 \end{align*}
Suppose that the input alphabet is $\Sigma$. This function is defined by an \mso interpretation as follows. The functor $F$ is $A^2$, and the universe formula is ``true'', which means that the positions of the output string are all pairs of positions in the input string. The order formula 
 describes the lexicographic order on $A^2$. Finally, the label of an output position is inherited from the input position on the second coordinate.
\end{myexample}

\subsection{List types}
\label{sec:list-types}
We are ultimately interested in functions that input and output strings over a finite alphabet. However, to create such functions using primes and combinators, it will be convenient to have more structured types for the simpler functions, such as pairs of strings. The idea to use such structured types comes from~\cite{bojanczykRegularFirstOrderList2018}, in particular we use the same types, as described in the following definition. 
\begin{definition}[List types]\label{def:string-type}
 A \emph{list type} is any type constructed using the constructors
 \begin{align*}
 \myunderbrace{1}{a type with \\ \scriptsize one element} \qquad 
 \myunderbrace{\Sigma_1 \times \Sigma_2}{pairs} \qquad \myunderbrace{\Sigma_1 + \Sigma_2}{co-pairs, i.e.\\ \scriptsize disjoint union} \qquad \myunderbrace{\Sigma^*}{lists}.
 \end{align*}
\end{definition}

A list type need not have $*$ as the topmost constructor, in fact it need not use $*$ at all. In the rest of this paper, we use $\Sigma$ and $\Gamma$ for list types, which may be infinite (unlike the convention in automata theory).
An example of a list type is 
\begin{align*}
(1+1+1)^*.
\end{align*}
This type can be seen as the type of strings over a three letter alphabet; in this way the list types generalize strings over finite alphabets. The generalization is minor, since elements of a list type can be seen as strings over a finite alphabet, which uses brackets and commas as in the following example:
\begin{align*}
\myunderbrace{([\leftinj 1, \rightinj 1, \leftinj 1], 1)
}{an element of the list type $(1+1)^* \times 1$}.\end{align*}

The type $1$ has a unique element, and therefore it admits a unique function from every other list type. In this sense, the type $1$ is a terminal object, assuming that morphisms are all functions. This continues to be true if we restrict the morhpisms to be the polyregular functions, see below.  However, in Section~\ref{sec:quantifier-free-functions} we will also consider a quantifier-free system, and in this system the unique function $1^* \to 1$ will not be allowed; in particular the type $1$ will no longer be a terminal object.

\paragraph*{Structures for list types.} We will be interested in \mso interpretations that transform one list type into another. We could simply represent list types as strings over a finite alphabet in the way described above, and then use \mso interpretations on strings over a finite alphabet. The resulting definition would be equivalent to the one that we will use in the paper. However, we choose to use a direct representation of list types as structures, without passing through a string encoding. The reason is that quantifiers would be needed to go between list types and their string encodings, and in this paper, we will be particularly interested in quantifier-free interpretations.

\begin{definition}\label{def:structures-for-string-types}
 To each list type we associate a class of structures, which is defined by induction as follows. 
 \begin{itemize}
 \item[($1$)] The class $1$ contains only one structure; this structure has one element in its universe and no relations.
 \item[($+$)]The vocabulary of the class $\Sigma_1 + \Sigma_2$ is the disjoint union of the vocabularies of the classes $\Sigma_1$ and $\Sigma_2$, plus one new nullary relation name (i.e.~arity zero). A structure in this class is obtained by taking a structure in either of the classes $\Sigma_1$ or $\Sigma_2$, extending the vocabulary to the vocabulary of the other class by using empty sets, and interpreting the new nullary relation as ``true'' or ``false'' depending on whether the structure is from $\Sigma_1$ or $\Sigma_2$.  The new nullary relation corresponds to the fact that co-pairs are tagged, i.e.~we know which of the two types $\Sigma_1$ or $\Sigma_2$ is used.
 \item[($\times$)] The vocabulary of the class $\Sigma_1 \times \Sigma_2$ is the disjoint union of the vocabularies of the class $\Sigma_1$ and $\Sigma_2$, plus one new unary relation name (i.e.~arity one). A structure in this class is obtained by taking the disjoint union (defined in the natural way) of two structures, one from $\Sigma_1$ and one from $\Sigma_2$, and interpreting  the new unary relation as the elements that come from the first structure. 
 \item[($*$)]The general idea is that a structure in the class $\Sigma^*$ is obtained by taking a list $[A_1,\ldots,A_n]$ of nonempty\footnote{A structure is nonempty if its universe is nonempty. This leads to the following subtle point, which arises when considering lists of lists, and related structures. Since a list can be empty, it follows that we do not allow lists of empty lists such as $[[],[],[]]$. This means that the list constructor, as it is used in this paper and formalized in Definition~\ref{def:structures-for-string-types}, should be interpreted as possibly empty lists with nonempty list items. This distinction will not play a role for types such as $(1+1)^*$ where list elements cannot be empty, which is the case that we really care about. } structures in $\Sigma$, creating a new structure using disjoint union (with a shared vocabulary), and adding a new binary relation $x \le y$ which holds whenever the structure containing $x$ appears earlier in the list (or in the same place) than the structure containing $y$. The problem with this construction is that it would mix nullary relations that come from different structures in the list. To fix this problem, each nullary relation name $R()$ in the vocabulary of $\Sigma$ is changed into a unary relation name $R(x)$ that selects elements $x$ such that the corresponding structure satisfies $R()$. 
 \end{itemize}
\end{definition}

If we apply the above representation to a list type 
\begin{align*}
(\myunderbrace{1+\cdots + 1}{$n$ times})^*
\end{align*}
then we get the representation of strings as ordered structures from Definition~\ref{def:string-as-structure}, with the exception that the empty string has a universe with one element. Therefore, it is not important if we use Definition~\ref{def:string-as-structure} or~\ref{def:structures-for-string-types} for representing strings.

\begin{definition}
 A \emph{polyregular function} is a function 
 \begin{align*}
 f : \Sigma \to \Gamma
 \end{align*}
 between list types that can be defined by an \mso interpretation, assuming that list types are viewed as classes of structures according to Definition~\ref{def:structures-for-string-types}.
\end{definition}

The original definition of polyregular functions~\cite{bojanczykPolyregularFunctions2018} did not use \mso interpretations, however \mso interpretations were shown equivalent to the original definition in~\cite[Theorem 7]{msoInterpretations}. Since the original definition was closed under composition, it follows that \mso interpretations are closed under composition (as long as the input and output classes are list types).

\section{The fold combinator}
\label{sec:string-types}
In this section, we discuss  dangers of the fold combinator
\begin{align*}
    \generalfold.    \end{align*}
We also explain how some of the dangers can be avoided by using quantifier-free interpretations. 

We begin this section with several examples illustrating the usefulness of fold.
\begin{myexample}
    \label{ex:groups-go-away}
    Consider a finite automaton with  $n$ states and an input alphabet of $m$ letters. Assuming some order on the states and alphabet, the transition function can be seen as a function between finite string types
    \begin{align*}
        {\myunderbrace{(1+ \cdots +1)}{$n$ times}} \times         {\myunderbrace{(1+ \cdots +1)}{$m$ times}} \to 
        \myunderbrace{1+ \cdots +1}{$n$ times}.
            \end{align*}
    If we apply fold to this automaton, under some chosen initial state, then we get the function that inputs a string, and returns the last state in the run. A special case of this construction is when both the states and input letters of the automaton are elements of some finite group $G$, the initial state is the group identity, and the transition function is the group operation. By folding this transition function, we get the \emph{group multiplication} function of type $G^* \to G$, which is one of the (less appealing) prime functions in the combinatory calculus from~\cite{bojanczykPolyregularFunctions2018}. 
\end{myexample}

\begin{myexample}\label{ex:identity-and-reverse} 
    There are two symmetric list constructors 
    \begin{align*}
       \myunderbrace{ 1 + \Sigma^* \times \Sigma \to \Sigma^*}{lists are constructed by adding \\ letters to the right of the list} 
       \hspace{2cm}
       \myunderbrace{ 1 + \Sigma \times \Sigma^* \to \Sigma^*}{lists are constructed by adding \\ letters to the left of the list}.
        \end{align*}
        If we apply fold to the two corresponding automata, then we get the   reverse and identity functions on lists, respectively. 
        The fold combinator corresponds in a canonical way to the first list constructor, which is why it is sometimes called \emph{fold right}.  
\end{myexample}



\subsection{On the dangers of folding}
\label{sec:dangers-of-fold}

We now present  two examples which show how the fold combinator, without any further restrictions, can define functions that are not polyregular.     More generally, one can use fold to derive any primitive recursive function~\cite[Section 4.1]{hutton1999tutorial}. In the examples below, the types $\Sigma$ and $\Gamma$ used by the fold combinator are infinite, since finite types would lead to polyregular functions, as explained in Example~\ref{ex:groups-go-away}.


\begin{myexample}[Iterating  duplication]\label{ex:fold-duplication}
    Consider an automaton where the input alphabet is $1$, and the states are  $1^*$. We view the states as natural numbers, with the list $1^n$ of length $n$ representing the number $n$. The  initial state in this automaton is  $1$, and the transition function is 
    \begin{align*}
    (1^n,1) \in 1^* \times 1 \quad  \mapsto \quad 1^{2n} \in 1^*.
    \end{align*}
This is an example of a polyregular function, in fact it is a linear regular function. However, if we apply fold to it, then we get the function 
\begin{align*}
    1^n \in 1^* \quad  \mapsto \quad 1^{2^n} \in 1^*.
    \end{align*}
 which is not  polyregular because of exponential growth. 
    \end{myexample}

\begin{myexample}[Subtraction]\label{ex:fold-tail}
    As illustrated in Example~\ref{ex:fold-duplication}, we run into trouble if we iterate duplication. But we can also run into trouble when the transition function does not create any new elements. Consider an automaton where the input alphabet is $1+1$, and the  state space is the integers, represented as the list type
    \begin{align*}
    \myunderbrace{1^*}{represents \\ 
    $\set{-1,-2,\ldots}$} 
    \quad + \quad 
    \myunderbrace{1^*}{represents \\ 
    $\set{0,1,\ldots}$} 
        \end{align*}
    The initial state is zero, and the transition function increments or decrements the state depending on which of the two input letters from $1+1$ it gets. This transition function is easily seen to be polyregular, and it has the property that the output size is at most the input size, assuming that the input letter contributes to the input size. However, by folding this automaton, we get a function that subsumes integer subtraction and is therefore not polyrergular. Using similar ideas, one could simulate two-counter machines. 
\end{myexample}

\subsection{Quantifier-free interpretations and their folding}

As the two above examples show, we have to be careful when applying fold. Clearly we must avoid duplication (Example~\ref{ex:fold-duplication}). This can be done by requiring the polynomial functor in the interpretation to be the identity, thus ensuring that the output is no larger than the input.  It is less clear how to avoid the problem with Example~\ref{ex:fold-tail}.  Our solution is to use {quantifier-free interpretations}, as defined below.

\begin{definition}
    A \emph{quantifier-free interpretation} is the special case of \mso interpretations where the polynomial functor is the identity $F(A)=A$ and all formulas are quantifier-free.
\end{definition}


One could consider interpretations in which the formulas are quantifier-free, but the functor is not necessarily the identity; such interpretations will not be useful in this paper.

The transition function in Example~\ref{ex:fold-tail} is not quantifier-free, since decrementing a number, which corresponds to removing a list element, is not a quantifier-free operation.
The following theorem is the first main contribution of this paper: fold can be safely applied to quantifier-free interpretations. In the theorem, a linear \mso interpretation is one that uses a functor $F$ that is linear in the natural sense.

\begin{theorem}\label{thm:fold-quantifier-free}    Let $\Sigma$ and $\Gamma$ be any classes of structures, not necessarily list types.    
If  the transition function
\begin{align*}
\delta : \Gamma \times \Sigma \to \Gamma
\end{align*}
in the assumption of the fold combinator is a quantifier-free interpretation, then the function in the conclusion is a linear \mso interpretation.
\end{theorem}
\begin{proof}
    Consider an automaton as in the assumption of the theorem.
    For an input to this automaton $[A_1,\ldots,A_n]$, and $i \in \set{0,\ldots,n}$ we write $B_i \in \Gamma$ for the state of the automaton after reading the first $i$ input letters. The state $B_0$ is the initial state, which is given by the assumption to the fold combinator, and the state $B_n$ is the last state, which is the output of the function in the conclusion of the fold combinator. Our goal is to compute the last state using a linear \mso interpretation. 
    
    Since the functor in $\delta$ is the identity, the output candidates are simply the elements of the input structure. Therefore,   the universe of $B_n$ is contained in the disjoint union of the universe of $B_{n-1}$ and the universe of $A_n$. By unfolding the induction, the universe of $B_n$ is contained in the universe of the first state $B_0$ and the input structure  $A=[A_1,\ldots,A_n]$. Therefore, to prove that the fold is an \mso interpretation, it will be enough to show that an \mso formula can tell us: (a) which elements of $B_0 + A$ belong to the output structure; and (b) which relations of the output structure are satisfied by which tuples from $B_0 + A$. The answers to these questions will be contained in the quantifier-free theory of the tuple, as defined below. 

    \begin{definition}\label{def:rank-free-theory} Let $A$ be a structure and let $\bar a$ be a list of distinguished elements, which need not belong to the universe of $A$. The \emph{quantifier-free theory} of a  $\bar a$ in $A$  is the following information:  which distinguished elements are in the universe, and which quantifier-free formulas  are satisfied by those distinguished elements that are in the universe. 
    \end{definition}


    Using the above terminology, to prove that the fold is definable in \mso, we need to show that for each tuple in $B_0 + A$, we can define in \mso the corresponding  quantifier-free theory in the output structure $B_0$. This will be done in the following claim. The key property used by the claim is the following \emph{continuity property} of quantifier-free interpretations: the quantifier-free theory of a tuple of output candidates in the output structure is uniquely determined by the quantifier-free theory of the same tuple in the input structure.

    In the following claim, we consider a function which inputs structures equipped with tuples of $k$ distinguished elements, and which has finitely many possible output values (quantifier-free theories, in the case of the claim). Such a function is called \mso definable if for every chosen output value, there is an \mso formula with $k$ free variables that selects inputs which give chosen output.
    \begin{claim}\label{claim:qf-folder} 
        For every $k \in \set{1,2,\ldots}$ and every tuple $\bar b$ of elements in $B_0$, the following 
        function is \mso definable:
        \begin{itemize}
        \item {\bf Input.} A structure $A \in \Sigma^*$ with  elements $\bar a \in A^k$.
        \item {\bf Output.} The quantifier-free theory of $\bar a \bar b$ in $B_n$.
        \end{itemize}
    \end{claim}
    \begin{proof}
        By the continuity property mentioned earlier in this proof, the quantifier-free theory of $\bar a \bar b$ in  $B_n$ is uniquely determined by the quantifier-free theory of $\bar a \bar b$ in the structure $(B_{n-1},A_n)$, which in turn is uniquely determined (by compositionality) by the quantifier-free theories of $\bar a  \bar b$ in the two individual structures  $B_{n-1}$ and $A_n$. Therefore, we can think of these quantifier-free theories as being computed by a finite automaton, where the initial state is the quantifier-free theory of $\bar b$ in  $B_0$, and   the input string is
        \begin{align*}
        [\text{qf theory of $\bar a$ in $A_1$}, \ldots, \text{qf theory of $\bar a$ in $A_n$}].
        \end{align*}
        By the continuity property, one can design a transition function for this automaton, which does not depend on the input structure $A$ or the tuple $\bar a$, such that its state after reading the first $i$ letters is the quantifier-free theory of $\bar a  \bar b$ in  $B_i$. 
        The state space of this automaton is finite, since there are finitely may quantifier-free theories once the vocabulary and number of arguments have been fixed.  Since finite automata can be simulated in \mso, it follows that the last state in the run of this automaton, which is the theory in the conclusion of the claim,  can be defined in \mso.
    \end{proof}
    
    We now use the claim to complete the proof of the lemma. The output candidates of the  \mso interpretation are defined by the polynomial  functor 
    \begin{align*}
        F(A)=A + \myunderbrace{1+ \cdots + 1}{size of initial state $B_0$}.
    \end{align*}
    In other words, the output candidates are elements of the input list and the initial state.  By the above claim, the quantifier-free theory of a single output candidate in the output structure can be defined in \mso, and since this theory tells us if the output candidate is present in the universe output structure, we can use it to define the universe.  Similarly, if we want to know if a tuple of output candidates satisfies some relation from the output vocabulary, then we can find this information using \mso as in the above claim.
\end{proof}

On its own, the theorem above does not solve all of the problems with fold. One issue is that the theorem only supports one application of fold, since the folded function is no longer quantifier-free and cannot be folded again. Another issue is that applying the theorem stays within the class of functions that do not increase the output size, while we will also be interested in folding functions that increase the size. These problems will be addressed later in the paper, by developing a suitable type system. 
Before continuing, we give some applications of the theorem.

\begin{myexample}\label{ex:finite-state-automaton-folded}
    Consider a transition function of a finite automaton as in Example~\ref{ex:groups-go-away}. In a list type of the form $1+ \cdots +1$, the component of the disjoint union that is used can be accessed by a quantifier-free formula without free variables, since it is represented using nullary relations. Therefore, the transition function is a quantifier-free interpretation, and so we can apply Theorem~\ref{thm:fold-quantifier-free} to conclude that the fold is an \mso transduction. This corresponds to the inclusion
\begin{align*}
\text{regular languages} \quad \subseteq \quad \text{\mso}.
\end{align*}
Applying Theorem~\ref{thm:fold-quantifier-free} to prove this inclusion is not the right way to prove it, since the inclusion itself is used in the proof of the theorem. 
\end{myexample}

In Example~\ref{ex:finite-state-automaton-folded}, we applied the fold combinator to a finite automaton. In the following example, we give a more interesting application, where the state space is infinite. 

\begin{myexample}[Streaming string transducers]\label{ex:sst} 
Define a \emph{simple streaming string transducer}, simple \sst for short,  as follows. It has two finite alphabets $\Sigma$ and $\Gamma$, called the \emph{input} and \emph{output} alphabets. It has a \emph{configuration space}, which is a list type of the form 
\begin{align*}
\Delta = (\Gamma^*)^{k_1} + \cdots + (\Gamma^*)^{k_m}.
\end{align*}
In other words, the set of configurations is obtained by applying some polynomial functor to the set of strings over the output alphabet. The idea is that a configuration consists of a state, which is one of the $m$ components, and a register valuation which is a tuple of strings over the output alphabet.
 The configurations of the transducer are updated according to the following three functions, which are required to be quantifier-free, according to the representation of the input and output alphabets that was used in Example~\ref{ex:finite-state-automaton-folded}:
\begin{align*}
    \myunderbrace{1 \to \Delta}{initial} 
    \qquad 
    \myunderbrace{\Delta \times \Sigma \to \Delta}{transition function} 
    \qquad 
    \myunderbrace{\Delta \to \Gamma^*}{final}.
    \end{align*}
The semantics of the transducer is the function of type $\Sigma^* \to \Gamma^*$ that is obtained by folding the first two functions, and post-composing with the final function.  
By Theorem~\ref{thm:fold-quantifier-free}, this function is an \mso transduction. 

The model described  is almost equivalent to in expressive power to the classical model of \sst~\cite[Section 3]{alurExpressivenessStreamingString2010}. The only difference, and the reason  why we call our model simple, is that our model allows the input letter to be used at most once (as opposed to a constant number of times) in the registers. The restriction on using each input letter being used at most once arises because we use a quantifier-free transition function, and such a function cannot duplicate letters.  The models would become equivalent if we could preprocess the input string by copying each input letter a constant number of times; in particular words every \sst can be decomposed as  a string-to-string homomorphism followed by a simple \sst.  Therefore, Theorem~\ref{thm:fold-quantifier-free} can be seen as subsuming the implication 
\begin{align*}
\text{\textsf{sst}} \subseteq \text{deterministic \mso transductions}
\end{align*}
proved in~\cite[Theorem 3]{alurExpressivenessStreamingString2010}. The same idea will work for trees, as we will see in Section~\ref{sec:trees}.
\end{myexample}

\begin{myexample}[Graphs]
    As mentioned in Theorem~\ref{thm:fold-quantifier-free}, the folded automaton need not operate on classes that are list types. For instance, we could adapt Example~\ref{ex:sst} to transducers in which the registers, instead of storing strings, store graphs with  $k$ distinguished vertices, as in Courcelle's algebras for treewidth~\cite[Section 1.4]{courcelleGraphStructureMonadic2012}. We could still apply Theorem~\ref{thm:fold-quantifier-free}, since the corresponding operations on graphs are quantifier-free, to prove that a graph extension of streaming string transducers~\cite[Section 3]{remarksGraphtoGraph} is subsumed by \mso transductions. Similar ideas would also work for cliquewidth. 
\end{myexample}

\section{Deriving quantifier-free functions}
\label{sec:quantifier-free-functions}
As we have shown in Theorem~\ref{thm:fold-quantifier-free}, the fold combinator can be safely applied to quantifier-free interpretations. Before discussing the fold combinator, we take a minor detour in this section, and present a complete system for the quantifier-free interpretations. 

\paragraph*{A few examples.} We begin with examples and non-examples of quantifier-free interpretations operating on list types.

\begin{myexample}[Commutativity of product]\label{ex:commutativity-of-product}
    Consider the function of type 
\begin{align*}
\Sigma_1 \times \Sigma_2 \to \Sigma_2 \times \Sigma_1,
\end{align*}
which swaps the order in a pair. Like all examples in this section, this   is actually an infinite family of functions, one for every choice of $\Sigma_1$ and $\Sigma_2$. The function is a quantifier-free interpretation. The only change between the input and output concerns the unary relation  from the definition of the product class $\Sigma_1 \times \Sigma_2$  which tells us if an element is from the first coordinate; this relation needs to be complemented.
\end{myexample}

\begin{myexample}[List reverse and concatenation]\label{ex:list-reverse}
Consider the list reverse function of type $\Sigma^* \to \Sigma^*$. This is clearly a quantifier-free interpretation -- it is enough to replace the order $x \le y$ with its reverse $y \le x$. A similar idea works for the list concatenation function of type $\Sigma^{**}  \to \Sigma^*$ which concatenates a list of lists into a list. In the input structure, there are two linear orders, corresponding to the inner and outer lists. To get the output structure, we use the lexicographic product of these two orders, which can be defined in a quantifier-free way.
\end{myexample}

\begin{myexample}\label{ex:terminal-object} As we mentioned in Section~\ref{sec:list-types}, the type $1$ is a terminal object if the morphisms are all functions, or the polyregular functions. However, it is no longer terminal when the morphisms are quantifier-free. This is because the unique function of type $1^* \to 1$ is not quantifier-free. The issue is that when the input is the empty list, the corresponding input structure has an empty universe, and therefore a quantifier-free function cannot create the one element in universe of  the ouput structure. A terminal object would be recovered by creating a new type, call it $0$, representing a class of structures that has only one structure, with an empty universe (not to be confused with the class that has no structures, which we do not consider). The corresponding prime functions for the type $0$ would be 
    \begin{align*}
     \primefunction{\Sigma}{\to}{\Sigma \times 0}{add $0$}\\
     \primefunction{0}{\to}{\Sigma^*}{create an empty list}.
    \end{align*}
    Note that we cannot have a version of ``add 0'' for 1, i.e.~a function of type $\Sigma \to \Sigma \times 1$; this is because creating the extra 1 would require unavailable resources. For similar reasons the function for creating empty lists that we use in Figure~\ref{fig:prime-quantifier-free} has type $\Sigma \to \Sigma \times \Gamma^*$ instead of the simpler type $1 \to \Gamma^*$; the latter function would be weaker, since it would use up more resources. All of these distinctions between $0$ and $1$ play a role only in the quantifier-free system; in the polyregular system the isomorphism $0 \leftrightarrow 1$ will be available.
\end{myexample}

\begin{myexample}[List constructor and destructor]\label{ex:list-add}
    Consider the (left) list constructor  
    \begin{align*}
1 +     \Sigma \times \Sigma^* \to \Sigma^*,
    \end{align*}
    that was discussed in Example~\ref{ex:identity-and-reverse}.  This is  a quantifier-free interpretation. If the input is from $1$, which can be tested in a quantifier-free way using the nullary relation from the co-product, then the output list is created in the natural way. Otherwise, if the input is a pair from $\Sigma \times \Sigma^*$, then the order on the concatenated list can easily be defined by using the unary predicate that identifies the first argument of a pair. 

    The list constructor is bijective, and therefore it has a corresponding inverse of type  
    \begin{align*}
        \Sigma^* \to 1 +     \Sigma \times \Sigma^*,
            \end{align*}
    which we call the \emph{list destructor}. The list destructor is not a quantifier-free interpretation. The reason is that if the input is an nonempty list, then we would need to isolate in a quantifier-free way the elements from the head, i.e.~from the first list element, which cannot be done.  
\end{myexample}


\begin{myexample}[Diagonal]\label{ex:diagonal-nonexample}
    Another non-example is  $x \mapsto (x,x)$. This is not a quantifier-free interpretation, since the output size is bigger than the input size.
\end{myexample}

\begin{figure}
    \begin{align*}
    \primeqf 
    \primefunction {\Sigma^* \times \Sigma}  \to {\Sigma^*} {append}\\
    \primefunction {\Sigma^{*}}  \to {\Sigma^*} {reverse}\\
    \primefunction {\Sigma^{**}}  \to {\Sigma^*} {concat}\\
    \primefunction {\Sigma}  \to {\Sigma \times \Gamma^*} {create empty} \\
    \primefunction {(\Sigma \times \Gamma)^*}  \to {\Sigma^* \times \Gamma^*} {list distribute}    \end{align*}
    \caption{\label{fig:prime-quantifier-free} The prime quantifier-free functions.}
\end{figure}

\begin{figure}
    \begin{align*}
        \combqf
 \frac{\Gamma \to \Sigma \quad \Sigma \to \Delta} { \Gamma \to \Delta} & & \text{function composition}
    \end{align*}
    \caption{\label{fig:combinator-quantifier-free} The quantifier-free combinators.}
\end{figure}


\paragraph*{A complete system.} We now present a complete characterization of quantifier-free interpretations on list types.  The system will be used as a basis for the system in the next section, which will describe general \mso interpretations.

\begin{theorem}\label{thm:qf-system}
    The quantifier-free interpretations between list types are exactly those that can be derived from  the prime functions in Figure~\ref{fig:prime-quantifier-free} by applying the combinators from  Figure~\ref{fig:combinator-quantifier-free}.
\end{theorem}

The proof of the above theorem, with completeness being the non-trivial part, is in the appendix.

\subsection{String diagrams} We conclude this section with several example derivations of quantifier-free functions using the system from Theorem~\ref{thm:fold-quantifier-free}. To present these derivations, we use  string\footnote{This is a name clash: the word ``string'' relates to the shape of the diagrams, and not to the fact that they manipulate types that represent strings.} diagrams based on~\cite[Chapter 3]{kissingerCoeckeQuantum}, as depicted in Figure~\ref{fig:binary-concatenation-string-diagram}.

\begin{figure}
    \mypic{4}
    \caption{\label{fig:binary-concatenation-string-diagram} A string diagram that derives the binary operation of type  $\Sigma^* \times \Sigma^* \to \Sigma^*$ for list concatenation.}
\end{figure}

We also  use string diagrams with a yellow background, where parallel wires represent co-products. For example, the following diagram represents the prime function from Figure~\ref{fig:prime-quantifier-free} that describes  commutativity of $+$:
\mypic{8}
Here are two other examples of string diagrams, which use dead ends, and represent projections and co-projections: 
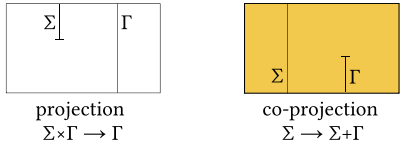
\mypic{9}

\begin{myexample}
    \label{ex:finite-sets} Recall the representation of finite sets as list types $1 + \cdots + 1$ used in Examples~\ref{ex:groups-go-away} and~\ref{ex:finite-state-automaton-folded}. Under this representation, every function between finite sets is derivable using the prime functions and combinators of Theorem~\ref{thm:fold-quantifier-free}.  This is easily seen using  string diagrams, as illustrated below:
    \mypic{10}
    The representation of finite sets as co-products is important here. For example, the diagonal function  $1 \to 1 \times 1$ is not derivable, as explained in Example~\ref{ex:diagonal-nonexample}. 
\end{myexample}
\section{Deriving polyregular functions}
We now move beyond quantifier-free functions and present the main contribution of this paper, which is a system that derives exactly the polyregular functions. As explained in Example~\ref{ex:fold-tail}, we cannot simply add the fold combinator to the system from Theorem~\ref{thm:fold-quantifier-free}.
Another idea would be to have two kinds of functions: quantifier-free functions, and general polyregular functions, with the fold combinator used to go from one kind to the other. In such a system, the only contribution of fold would be to define linear regular functions, since such are the functions in the conclusion of Theorem~\ref{thm:fold-quantifier-free}. We are more ambitious, and we want the fold combinator to be useful also for non-linear functions.

To define a system with fold, we add a new unary type constructor. This type constructor is denoted by $!$ and it is written on the left. The general idea is that an element $!x$ is essentially the same element as $x$, except that it is harder to obtain.  The type constructor is not idempotent, and so $!!x$ is even harder to obtain than $!x$. The goal of this type constructor  is to restrict the application of fold in a way that avoids the problems discussed in Section~\ref{sec:dangers-of-fold}. This is done by using the following \emph{safe fold} combinator: 
\begin{align*}
    \combinator{ !^k 1 \to \Gamma \quad \Gamma \times \Sigma \to \Gamma}{ !^k(\Sigma^*) \to \Gamma} {safe fold}
\end{align*}
In the combinator, $!^k$ refers to $k$-fold application of $!$. When applying the combinator, the number   $k \in \set{0,1,\ldots}$ must be strictly bigger than the grade of $\Gamma$, which is defined to be the maximal nesting of $!$, as in the following examples:
\begin{align*}
\myunderbrace{1^*}{grade zero} \qquad \myunderbrace{1+!(1 + !1)}{grade two}.
\end{align*}
For example, when $\Gamma$ has grade zero, i.e.~it does not use $!$, then safe fold can be used in the form 
\begin{align*}
    \combinator{ ! 1 \to \Gamma \quad \Gamma \times \Sigma \to \Gamma}{ !(\Sigma^*) \to \Gamma} {safe fold when $\Gamma$ is without $!$}
\end{align*}

The general idea is that the annotation with $!$  will disallow certain kinds of repeated applications of fold that would lead to functions that are not polyregular. 
Before giving a formal description of the system, we begin with an example.


\begin{myexample}[List destructor]\label{ex:deconstructor}
    In this example, we use safe fold to derive a variant of the  list destructor 
    \begin{align*}
        \Sigma^* \to  1 + \Sigma^* \times \Sigma
        \end{align*}
    that was discussed in Example~\ref{ex:list-add}.  Consider an automaton where the state space is the output type of the list destructor, 
    the initial state is $1$, and the transition function is 
    \mypic{12}
    By applying the safe fold to this automaton, we get the list deconstructor in  a weaker type, namely   
    \begin{align*}
        !(\Sigma^*) \to  1 + \Sigma^* \times \Sigma.
        \end{align*}
The weaker type avoids the issues from Example~\ref{ex:fold-tail}, since the input and output will have different numbers of $!$, and therefore we will be unable to apply fold again. 
\end{myexample}


\subsection{Graded types and their derivable functions} 
We now give a  formal description of the system. The type system is the same as previously, except that we have one more type constructor for $!$.

\begin{definition}\label{def:graded-string-type}
    A \emph{graded list type} is any type that is constructed using the following type constructors
    \begin{align*}
        \myunderbrace{1}{a type with \\ \scriptsize  one element} \qquad 
        \myunderbrace{\Sigma_1 \times \Sigma_2}{pair} \qquad \myunderbrace{\Sigma_1 + \Sigma_2}{co-pair, i.e.\\ \scriptsize disjoint union} 
        \qquad \myunderbrace{\Sigma^*}{lists}
        \qquad !\Sigma.
        \end{align*}
\end{definition}

The general idea is that $!$ does not change the underlying set, but only introduces some type annotation that controls the way fold and duplication can be applied. Apart from safe fold, the main way of dealing with $!$ is the duplicating operation
\begin{align*}
    \primefunction{\ \upgrade{\Sigma}}{\rightarrow}{\upgrade \Sigma \times \Sigma  }{{absorption},}
\end{align*}
which is named after the same rule in the parsimonious calculus  of Mazza~\cite[p.1]{mazza2015simple}. There are also prime functions for commuting $!$ with the remaining type constructors, for example $![x,y,z]$ and $[!x,!y,!z]$ are going to be equivalent in our system; for this reason we can write $!\Sigma^*$ without specifying the order in which the two constructors are applied.

\begin{definition} There are two kinds of derivability for functions between graded list types. 
    \begin{enumerate}
        \item  \textbf{Strongly derivable.} A function is called    \emph{strongly derivable} if it can be derived using the quantifier-free prime functions and combinators from Figures~\ref{fig:prime-quantifier-free} and~\ref{fig:combinator-quantifier-free}, extended to graded list types that can use $!$, along with four new prime functions
        \begin{align*}
            \primefunction{\ \upgradelong{\Gamma + \Sigma}}{\leftrightarrow}{\upgrade \Gamma + \upgrade  \Sigma }{{! commutes with $+$}}\\
            \primefunction{\ \upgradelong{\Gamma \times \Sigma}}{\leftrightarrow}{\upgrade \Gamma \times \upgrade  \Sigma }{{! commutes with $\times$}}\\
            \primefunction{\ (\upgrade \Gamma)^*}{\leftrightarrow}{\upgradelong{\Gamma^*}}{{! commutes with $*$}}\\
            \primefunction{\ \upgrade{\Gamma}}{\rightarrow}{\upgrade \Gamma \times \Gamma  }{{absorption}}\end{align*}
        and two new combinators
        \begin{align*}
            \combinator{\Sigma \to \Gamma}{!\Sigma \to !\Gamma}{functoriality of !}\\
            \combinator{ !^k 1 \to \Gamma \quad \Gamma \times \Sigma \to \Gamma}{ !^k(\Sigma^*) \to \Gamma} {safe fold} 
        \end{align*}
        The safe fold combinator can only be applied when $\Gamma$ has grade $<k$.
        \item \textbf{Weakly derivable.} A function is called \emph{weakly derivable} if it can be decomposed as 
        \[
        \begin{tikzcd}
            [column sep = 1.5cm]
        \Sigma 
        \ar[rr,"\text{ $k$-fold application of $!$}"]
        &&
        !^k \Sigma 
        \ar[r,"f"]
        &
        \Gamma
        \end{tikzcd}
        \]
        for some $k \in \set{0,1,\ldots}$ and strongly derivable $f$. 
    \end{enumerate}

\end{definition}

In other words, a function is weakly derivable if it can be strongly derived for a sufficiently upgraded input type. For example, the list destructor of type 
\begin{align*}
    \Sigma^* \to  1 + \Sigma^* \times \Sigma
\end{align*}
function is not strongly derivable (Example~\ref{ex:list-add}), but it is weakly derivable (Example~\ref{ex:deconstructor}). 

In the following theorem, which is the main result of this paper, we are only interested in weak derivability for functions between (ungraded) string types, i.e.~between types that do not use $!$. The purpose of ! is to get the strong derivations.

\begin{theorem}\label{thm:main} A function  between (ungraded) list types
    is polyregular if and only if it is weakly derivable.
\end{theorem}

The proof has two parts: soundness and completeness.

\subsection{Completeness}
The completeness part of Theorem~\ref{thm:main} is that every polyregular function can be weakly derived. Unlike the quantifier-free system in Theorem~\ref{thm:qf-system}, completeness   is relatively easy. This is because fold is a powerful combinator, and we can draw on a prior complete system for the polyregular functions~\cite[p.~64]{bojanczykPolyregularFunctions2018}. In the completeness proof, the polynomial growth output size will come from a single quadratic function.

\begin{claim}\label{claim:reverse-prefix}
    One can weakly derive the following function 
    \begin{align*}
        \myunderbrace{[a_1,\ldots,a_n] \quad \mapsto \quad  [[a_n,\ldots,a_1],[a_{n-1},\ldots,a_1],\ldots,[a_1]]}{call this the \emph{prefixes} function}. \end{align*}
\end{claim}
\begin{proof}
    Consider an automaton, where the input alphabet is $!\Sigma$, the state space is $\Sigma^{**} \times !\Sigma^*$, the initial state is the pair of empty lists, and the transition function is 
    \mypic{3}
    By applying fold to this automaton, we get a function of type 
    \begin{align*}
    !!\Sigma^* \to \Sigma^{**} \times !\Sigma^*
    \end{align*}
    which returns the output of the prefixes function on the first output coordinate. Observe that in this proof, we applied the fold to a transition function that already uses $!$.
\end{proof}

Using the above function, in the appendix we show that the weakly derivable functions contain an already existing complete system for the polyregular functions~\cite[p.~64]{bojanczykPolyregularFunctions2018}. 

Before discussing the soundness proof in the theorem, let us comment on the minimality of its system. The system inherits all of the primes and combinators from the quantifier-free system in Theorem~\ref{thm:qf-system}. In the presence of fold, some of these primes and combinators can be derived thus leading to a smaller system.

\begin{theorem}\label{thm:reduced}
    The system from Theorem~\ref{thm:main} remains complete after removing the map combinator, as well as all prime functions and combinators that involve the list type, and adding 
    \begin{align*}
        \primefunction {1 + \Sigma} \to {\Sigma^*} {lists of length at most one}\\
        \primefunction {\Sigma^* \times \Sigma^*} \to {\Sigma^*} {binary list concatenation.}
        \end{align*}       
\end{theorem}

\subsection{Soundness}
\label{sec:soundness}
The rest of this section is devoted to the proof of soundness for Theorem~\ref{thm:main}, which is that all weakly derivable functions are polyregular. We will define an invariant on strongly derivable functions, which is satisfied by the prime functions, is preserved by the combinators, and which implies that a function is polyregular. This invariant can be seen as giving a semantic explanation of the $!$ constructor and the strongly derivable functions.

The invariant uses a more refined notion of \mso interpretations, called \emph{graded \mso interpretations}. These interpretations operate on graded structures, as described in the following definition.

\begin{definition}[Graded structure]\label{def:graded-structure}
    A \emph{graded structure} is a structure, together with a \emph{grading function} that assigns to each element in the universe a grade in $\set{0,1,\ldots}$.
\end{definition}

The idea is that the grade of an element is the number of times that $!$ has been applied, as in the following example
\begin{align*}
(\myunderbrace{1}{grade\\ zero},\myunderbrace{![1,1,1]}{grade \\ one}).
\end{align*}
A graded list type can be seen as describing a class of graded structures, with the constructor $!$ incrementing the grade of all elements, and the remaining constructors treated in the same way as in Definition~\ref{def:structures-for-string-types}.

If $A$ is a graded structure,  we write $A|\ell$ for the structure that is obtained from $A$ by restricting its universe to elements that have grade at least $\ell$. In the definition of a graded \mso interpretation, we use the grades to control how an \mso interpretation $f$ uses quantifiers. The general idea is that $f(A)|\ell$ depends on $A|\ell$ in a quantifier-free way, and on $A|\ell+1$ in an \mso definable way. 

Before presenting the formal definition, we introduce some notation, in which a polynomial functor $F$ is applied to a tuple of elements $\bar a$, yielding a new (typically longer) tuple of elements $F(\bar a)$. If an input set $A$ for  a polynomial functor $F$ is equipped with some linear order, then this linear order can be extended to a linear order on the output set $F(A)$, by using some fixed order on the components, and ordering tuples lexicographically. This way we can think of a polynomial functor as transforming linearly ordered sets, i.e.~lists. We will care about lists of fixed length, which we call tuples. For example if the polynomial functor is   $A  + A^2$, then applying it to the tuple $(1,2)$ gives the tuple 
\begin{align*}
(1,2,1,2, (1,1), (1,2), (2,1), (2,2)) \in F(\set{1,2})^6.
\end{align*}  
In the definition below, we will care about the theories of tuples of the form $F(\bar a)$, with the theories defined as in Definition~\ref{def:rank-free-theory}, but extended to \mso formulas of given quantifier rank (the quantifier rank of an \mso formula is the nesting depth of the quantifiers, with first-order and second-order quantifiers counted in the same way). Recall that these theories allow for distinguished elements that are not part of the universe in a structure.  Equipped with this notation, we are ready to define the graded version of \mso interpretations.

\begin{definition}\label{def:graded-interpretation}
    A function $f : \Sigma \to \Gamma$ is called a \emph{graded \mso interpretation} if there is some polynomial functor 
    \begin{align*}
    F(A) \quad =\quad \myunderbrace{A}{this is called the \\ \emph{quantifier-free} \\ component} \quad +\quad   \myunderbrace{F_0(A) + \cdots + F_m(A)}{components from this part \\ of the functor are called the \\ \emph{downgrading} components}  
    \end{align*}
 such that the following conditions hold: 
    \begin{enumerate}
        \item {\bf Universe and grades.} The universe of the output structure is contained in 
        \begin{align*}
        A + F_0(A|1) + F_1(A|2) + \cdots + F_m(A|m+1).
        \end{align*}
        The grades in the output structure are defined as follows: elements from $F_\ell$ have grade $\ell$, and elements from the quantifier-free component inherit their grade from $A$.
        \item {\bf Continuity.} For every $k,\ell\in \set{0,1,\ldots}$ there is some quantifier rank $r \in \set{0,1,\ldots}$ such that for every input structure $A$  and distinguished elements  $\bar a \in A^k$, 
 the quantifier-free theory of the tuple $F(\bar a)$ in $f(A)|\ell$  is uniquely determined by the following two theories: 
        \begin{enumerate}
            \item the quantifier-free theory of $\bar a$  in $A|\ell$;
            \item the rank $r$ \mso theory of $\bar a$ in $A|\ell+1$.
        \end{enumerate}
    \end{enumerate}
\end{definition}

If we ignore the grades, then a graded \mso interpretation is a special case of an \mso interpretation. This is because the types mentioned in the continuity condition will tell us which output candidates from $F(A)$ are in the universe of the output structure, and how the relations of the output structure are defined on them. Therefore, the continuity condition tells us that the output can be defined in \mso, and even in a way that respects the grades. 

Conversely, we can also view each (ungraded) \mso interpretation as a graded \mso interpretation with a trivial grade structure: namely all input elements have nonzero grade, say grade one, and all output elements have zero grade. With such a trivial grade structure, the continuity condition collapses to the usual condition in an \mso interpretation. 

Graded \mso interpretations also generalize quantifier-free interpretations -- this happens in the case when all elements in the input and output structures have grade zero. In this case, only the quantifier-free component is useful, and all formulas are quantifier-free. 

In the appendix, we show that 
 all strongly derivable prime functions are graded \mso interpretations. This will imply that all weakly derivable functions are ungraded \mso interpretations, since the continuity condition becomes vacuous when the input type is sufficiently upgraded. 
The proof is an induction on the size of a strong derivation, with the most interesting cases being composition and safe fold. Composition is a corollary of composition closure for \mso interpretations on string types~\cite[Corollary 8]{msoInterpretations}, while safe fold is treated in the same way as in Theorem~\ref{thm:fold-quantifier-free}.

\section{Linear regular functions}
\label{sec:linear-regular-functions}
The last group of results from this paper concerns the linear regular functions, i.e.~polyregular functions of linear growth. We show that a small change to the system from Theorem~\ref{thm:main} will give exactly the linear regular functions. As we will see, superlinear growth in the system from Theorem~\ref{thm:main} is not created by the fold combinator, with the culprit instead being 
\begin{align*}
    \primefunction{\ \upgrade{\Gamma}}{\rightarrow}{\upgrade \Gamma \times \Gamma  }{\ \quad \quad {absorption}}.
\end{align*}
This function allows us to create an unbounded number of copies of an element of $\Gamma$, as witnessed in the proof of Claim~\ref{claim:reverse-prefix}. If we simply remove this function, then the system will become too weak, since all other prime functions and combinators preserve the property that the universe of the output structure is contained in the universe of the input structure. The solution is to add a weaker form of absorption 
\begin{align*}
    \primefunction{\ \upgrade{\Gamma}}{\rightarrow}{\Gamma \times \Gamma  }{linear absorption}.
\end{align*}
In other words, removing \emph{all} occurrences of $!$ is the price paid for copying. The corresponding system describes exactly the linear regular functions, as stated in the following theorem.

\begin{theorem}\label{thm:linear}
    A function 
    $
        f : \Sigma \to \Gamma
     $ between string types
        is linear regular if and only if it can be weakly derived in a system that is obtained from the one in Theorem~\ref{thm:main}\footnote{One can also start with the smaller system from Theorem~\ref{thm:reduced}.} by replacing  absorption with linear absorption.\end{theorem}

        The proof for the above theorem, which is in the appendix, is based on Example~\ref{ex:sst} about streaming string transducers. The idea is that linear absorption together with fold is enough to simulate streaming string transducers, which are expressively complete the linear regular functions.

\subsection{Tree types}
\label{sec:trees}
It turns out that the system for linear regular functions from Theorem~\ref{thm:linear} can be generalized without much further difficulty to trees. This is in contrast to a prior combinator system for trees~\cite[Theorem 7.1]{bojanczykDoumane2020}, which had an involved proof using approximately fifty prime functions. We believe that this is evidence for the usefulness of the fold combinator.

Consider a type for trees, defined inductively by
\begin{align*}
\trees \Sigma = \myunderbrace{1 + \trees \Sigma  \times \Sigma \times \trees \Sigma}{a tree is either a leaf, or has two \\ subtrees and a root label}
\end{align*}
A \emph{tree type} is a type that is constructed using the types from Definition~\ref{def:string-type}, together with the tree type. Tree types can be seen as structures, using  the same construction as for lists in Defintion~\ref{def:structures-for-string-types}, except that instead of one linear order, we have two orders:  the \emph{descendant order}   (which is not a linear order) and the \emph{document order} (aka infix order) given by 
\begin{align*}
\text{left subtree} 
\quad < \quad 
\text{root}
\quad < \quad 
\text{right subtree.}
\end{align*}
 Define a \emph{linear regular tree function} to be a function between tree types that is defined using linear \mso transductions. 

Following Wilke~\cite{wilke1996algebraic}, we view trees as an algebra. In this  algebra, there is an additional type constructor $\contexts \Sigma$, which describes \emph{contexts}. A context is a tree with a distinguished leaf (called the \emph{hole}) where other trees can be inserted. This is not a primitive type constructor, only syntactic sugar for a certain combination of the list and tree type constructurs:
    \begin{align*}
    \contexts \Sigma \eqdef (\myunderbrace{(\trees \Sigma \times \Sigma)}{the hole is in \\  the right subtree}\ +\  \myunderbrace{( \Sigma \times \trees \Sigma)}{the hole is in \\  the left subtree})^*.
    \end{align*}
To operate on trees and contexts, we use the following operations,  called \emph{Wilke's operations}, see~\cite[Figure 1]{wilke1996algebraic}: 
\begin{align*}
    \primefunction{1 + \trees \Sigma  \times \Sigma \times \trees \Sigma}{\to}{\trees \Sigma}{tree constructor}\\
    \primefunction{\contexts \Sigma \times \trees \Sigma}{\to}{\trees \Sigma}{replace hole by a tree}\\
    \primefunction{\contexts \Sigma \times \contexts \Sigma}{\to}{\contexts \Sigma}{context composition}\\
    \primefunction{1 + (\trees \Sigma \times \Sigma) + (\Sigma \times \trees \Sigma)}{\to}{\contexts \Sigma}{context creation}\\
\end{align*}
All of these operations are quantifier-free interpretations, and we will use them as primes. The last two operations need not be explicitly added, since they can  derived using the system from Theorem~\ref{thm:fold-quantifier-free}.

\begin{theorem}\label{thm:linear-trees}
    A function 
    $
        f : \Sigma \to \Gamma
     $ between tree types
        is linear regular if and only if it can be derived in a system that is obtained from the system in Theorem~\ref{thm:linear} by adding the tree type, Wilke's operations, the prime function
        \begin{align*}
        \primefunction{!\trees \Sigma}{\leftrightarrow}{\trees !\Sigma}{! commutes with $\trees$}
        \end{align*}
        and  the following combinator
        \begin{align*}
            \frac{!^k1 \to \Gamma \quad \Gamma \times \Sigma \times  \Gamma \to \Gamma}{!^k \trees \Sigma \to \Gamma} & & \text{safe tree fold,}
        \end{align*}
        which can be applied whenever $\Gamma$ has grade $<k$.
\end{theorem}
\begin{proof}[Sketch]
    As in Theorem~\ref{thm:linear}. We use the same soundness proof, except that tree automata are used instead of string automata. For completeness, we use a result of Alur and D'Antoni, which says that every linear \mso interpretation is computed by a streaming tree transducer~\cite[Theorem 4.6]{alur2014regular}. Adjusting for notation, a streaming tree transducer is defined in the same way as in Example~\ref{ex:sst}, except that instead of lists, registers store trees and contexts. The registers in the transducer are manipulated using Wilke's operations; and thus for the same reason as in Example~\ref{ex:sst}, the corresponding tree function is weakly derivable. This completeness proof takes into account only functions of type $\trees \Sigma \to \trees \Gamma$ where $\Sigma$ and $\Gamma$ are finite alphabets, but the extension to other tree types is easily accomplished by encoding tree types into such trees.
\end{proof}

\paragraph*{Tree polyregular functions.}
It is natural to ask about a polyregular system for trees. We conjecture that if we add absorption to the system from Theorem~\ref{thm:linear-trees}, and possibly a few extra  prime functions, then the system will define exactly the \mso interpretations on tree types. This conjecture would imply that tree-to-tree \mso inprepretations are closed under composition, which is an open problem.

\section{Perspectives}
We finish the paper with some directions for future work. 

In our proofs, we are careless about the number of times that $!$ is applied. Maybe a more refined approach can give a better understanding of the correspondence between the nesting of $!$ and the resources involved, such as quantifiers or copying. 
Alternatively, one could try to do away with $!$ entirely, and use some proof system where the safety of fold is captured by a structural property of the proof. One idea in this direction is to look at cyclic proofs~\cite{brotherston2011sequent}. Another idea would be to capture the structural property using the visual language of string diagrams.

Another question that concerns string diagrams is about the equivalence problem.
Decidability of the equivalence problem for polyregular functions is an open problem, but in the case of linear functions the problem is known to be decidable~\cite[Theorem 1]{gurariEquivalenceProblemDeterministic1982}. Maybe one can express the decision procedure in terms of string diagrams, by  designing equivalences on string diagrams which identify exactly those diagrams that describe the same function.

The system in this paper is based on combinators. A more powerful system would also allow for variables, $\lambda$, and higher-order types. Such a system exists without fold~\cite[Section 4]{polyregular-survey}, and it is tempting to see if it can be extended with fold. This extension is not guaranteed to work, since  there are examples of higher-order linear calculi where the corresponding complexity is super-polynomial, in fact primitive recursive, e.g.~\cite[Theorem 5]{lago2005geometry} or~\cite[Theorem 2.15]{kuperberg2021cyclic}. If successful, the extension would  be an expressive functional programming language that can only define regular functions.

\bibliographystyle{plain}
\bibliography{bib}

\appendix
\section{The quantifier-free system}
\label{sec:qf-system}
In this part of the appendix, we prove Theorem~\ref{thm:qf-system}. In the proof, a \emph{derivable function} is a function that can be derived using the system from Theorem~\ref{thm:qf-system}. In other parts of the paper, derivable functions will refer to other systems. 

The proof of Theorem~\ref{thm:qf-system} has two parts: soundness (i.e.~all derivable functions are quantifier-free interpretations) and completeness (i.e.~all quantifier-free interpretations are derivable).

\subsection{Soundness}
\label{sec:qf-soundness}
To prove soundness of the system, we show that all prime functions from Figure~\ref{fig:prime-quantifier-free} are quantifier-free interpretations, and that the class of quantifier-free interpretations is closed under applying all combinators from Figure~\ref{fig:combinator-quantifier-free}. 

Several of the prime functions from Figure~\ref{fig:prime-quantifier-free} where already shown to be quantifier-free interpretations in Examples~\ref{ex:commutativity-of-product}, \ref{ex:list-reverse} and~\ref{ex:list-add}. The remaining prime functions are left to the reader. 

Let us now consider the combinators from Figure~\ref{fig:combinator-quantifier-free}.

The first observation is that quantifier-free interpretations are closed under composition. This is because: (a) the functor in a quantifier-free interpretation is the identity functor, and composing this functor with itself gives the same functor; and (b) quantifier-free formulas are closed under substitution. 

The other combinators are easiily seen to preserve quantifier-free interpretations. We only discuss one case in more detail, namely the combinator 
\begin{align*}
\combinator{\Sigma \to \Gamma}{\Sigma^* \to \Gamma^*}{functoriality of $*$,}
\end{align*}
which is also known as the \emph{map combinator}. 
The difficulty with this combinator is that in the structure that represents a list of elements $[A_1,\ldots,A_n] \in \Sigma$, as per Definition~\ref{def:structures-for-string-types}, the nullary predicates from the structures $A_1,\ldots,A_n$ are replaced by unary predicates. However, since the same replacement is done for the output list, it follows that a straightforward syntactic construction can be applied to transform the quantifier-free interpretation from the assumption of the combinator into a quantifier-free interpretation from the conclusion.

\subsection{Completeness}
\label{sec:qf-completeness}
The rest of this section is devoted to the completeness proof. We begin with some notation and preparatory lemmas that will be used in the proof.

\paragraph*{Zero type.} We will use an extended system, which has an additional type called $0$, as discussed in Example~\ref{ex:terminal-object}. This type represents a class that contains one structure, and that structure has an empty universe. (This class is terminal, in the sense that every class of structures admits a unique quantifier-free interpretation to $0$.) The corresponding prime functions are 
\begin{align*}
 \primefunction{\Sigma}{\to}{\Sigma \times 0}{add $0$}\\
 \primefunction{0}{\to}{\Sigma^*}{create an empty list}
\end{align*}
One should not confuse $0$ with the empty class $\emptyset$ (which anyway is not part of our type system). For example, 
\begin{align*}
0 + \Sigma \neq \Sigma = \emptyset + \Sigma.
\end{align*}

The extended system with $0$ is equivalent to the original system, since we can view $0$ as $1^*$, but with only the empty list used. In particular, the extended system is conservative in the following sense: if a function between types that do not use $0$ is derivable in the extended system, then it is also derivable in the non-extended system. For this reason, we can do the completeness proof in the extended system, which will be slightly more convenient. From now on, list types can use $0$.

\paragraph*{Disjunctive normal form.}
It will be useful to consider list types in a certain normal form, which is achieved using distributivity. We say that a list type is in \emph{disjunctive normal form} if it is of the form 
 \begin{align*}
 \coprod_{i \in I} \prod_{j \in I_j} \Sigma_{i,j}
 \end{align*}
 where each $\Sigma_{i,j}$ is one of the types $0$ or $1$, or a list $\Sigma^*$ where $\Sigma$ is in disjunctive normal form. In other words, the list type does not contain any product of co-products.
 
 In our proof, the main advantage of this normal form concerns nullary relations. Recall that the nullary relations in Definition~\ref{def:structures-for-string-types}, appear only in the co-product, and they are removed when applying the list constructor. Therefore, if a type in disjunctive normal form is not a co-product type, then its vocabulary contains no nullary relations.

 The following lemma shows that every list type admits a derivable isomorphism with some list type in disjunctive normal form. Here, a \emph{derivable isomorphism} is a derivable function that has a derivable inverse. 
 \begin{lemma}\label{lem:dnf}
 Every list type admits a derivable isomorphism with some list type in disjunctive normal form. 
 \end{lemma}
 \begin{proof}
 Using distributivity and functoriality. 
 \end{proof}

 Thanks to the already proved soundness part of the theorem, the derivable isomorphism is also quantifier-free. Therefore, to prove completeness of the system, it is enough to prove completeness only for functions where both the input and output types are in disjunctive normal form. From now on, we only consider list types in disjunctive normal form.

\paragraph*{Safe pairing.} The last issue to be discussed before the completeness proof concerns pairing functions. Suppose that 
\begin{align*}
f : \Sigma \to \Gamma_1 \times \Gamma_2
\end{align*}
is a quantifier-free interpretation. In the completeness proof, we will want to show that it is derivable. A natural idea would be to use an inductive argument to derive the two quantifier-free interpretations 
\begin{align*}
f_i : \Sigma \to \Gamma_i
\end{align*}
that arise from $f$ by projecting it onto the two output coordinates, and to then pair these two derivations into a derivation of $f$. Unfortunately, combining these two derviations would require some kind of pairing combinator, or a duplicating function of type $\Sigma \to \Sigma \times \Sigma$, none of which are available in our system (because they would be unsound). 

For these reasons, we need to be a bit careful with pairing. The crucial observation is that pairing is not always unsound, because some functions can be paired. For example, the two functions $f_1$ and $f_2$ described above can be paired, because they use disjoint parts of the input structure. More formally, the universe formulas are disjoint, i.e.~no element can be selected by both universe formulas. This view will be used in the completeness proof. To formalize it, we use the following lemma.

\begin{lemma}\label{lem:safe-pairing}
 Let $\Sigma$ be a list type in disjunctive normal form, and let $\varphi(x)$ be a quantifier-free formula over its vocabulary. There is a list type, denoted by $\Sigma|\varphi$, and a quantifier-free interpretation
 \[
 \begin{tikzcd}
 [column sep =2cm]
 \Sigma \ar[r,"\text{projection of $\varphi$}"]
 &
 \Sigma|\varphi
 \end{tikzcd}
 \]
 such that the following conditions are satisfied.
 \begin{enumerate}
 \item \label{it:safe-pairing-factors} For every quantifier-free interpretation $f : \Sigma \to \Gamma$, such that the universe formula of $f$ is contained in $\varphi$ (which means that the universe formula of $f$ implies the formula $\varphi$), there is a 
 decomposition 
 \[
 \begin{tikzcd}
 \Sigma 
 \ar[rr,"f"]
 \ar[dr,"\text{projection of $\varphi$}"']
 && \Gamma\\
 & \Sigma | \varphi
 \ar[ur,"f|\varphi"']
 \end{tikzcd}
 \]
 where $f|\varphi$ is a quantifier-free interpetation.
 \item \textbf{Safe pairing.} Suppose that $\varphi_1,\ldots,\varphi_n$ are formulas as in the assumption of the lemma, which are pairwise disjoint. Then one can derive the function 
 \[
 \begin{tikzcd}
 \Sigma \ar[r] 
 &
 (\Sigma | \varphi_1) \times \cdots \times (\Sigma| \varphi_n)
 \end{tikzcd}
 \]
 that produces all projections in parallel.
 \end{enumerate}
 \end{lemma}
\begin{proof} 
    The purpose of the type $0$ is in this lemma. The type $0$ is used for $\Sigma|\varphi$ when the formula $\varphi(x)$ selects no elements. The lemma is proved by induction on the structure of the type $\Sigma$. 
 \begin{itemize}
 \item Suppose that $\Sigma$ is the zero type $0$. In this case, the formula $\varphi$ must be equivalent to ``false''. We define $0|\varphi$ to be the same type $0$, and the projection is the identity. The safe pairing condition holds because of the prime function $\Sigma \to \Sigma \times 0$.
 \item Suppose that $\Sigma$ is the unit type $1$. In this case, the formula $\varphi$ is equivalent to either ``false'' or ``true'', since the unique structure in $1$ has a universe that has only one element. We define $1|\varphi$ to be the $0$ or $1$, depending on which of the two cases holds, with the projection being the unique function $1 \to 1|\varphi$. The safe pairing condition is proved using the prime function $\Sigma \to \Sigma \times 0$, since the list of quantifier-free formulas in the condition can have at most one formula that is not ``false''.\item Consider a list type of the form $\Sigma^*$. 
 The main observation in the proof is that for quantifier-free formulas $\varphi(x)$ with one free variable, there is a bijective correspondence between formulas over the vocabularies of $\Sigma$ and $\Sigma^*$. This correspondence is defined as follows: for every formula $\varphi$ over the vocabulary of $\Sigma$ with one free variable, there is a formula $\varphi^*$ over the vocabulary of $\Sigma^*$ such that for every list 
 \begin{align*}
 A= [A_1,\ldots,A_n] \in \Sigma^*,
 \end{align*}
 an element $a \in A_i$ is selected by $\varphi^*$ in the entire list $A$ if and only if $a$ is selected by $\varphi$ in the list element $A_i$. It is not hard to see that such a formula exists, and furthermore, every formula over the vocabulary of $\Sigma^*$ is of equivalent to a formula of the form $\varphi^*$. 

 Therefore, in the case when the type is a list $\Sigma^*$, we can assume that the formula over the vocabulary of $\Sigma^*$ is of the form $\varphi^*$ for some formula $\varphi$ over the vocabulary of $\Sigma$. Define 
 \begin{align*}
 \Sigma^*|\varphi^* \quad \eqdef \quad (\Sigma|\varphi)^*,
 \end{align*}
 with the projection function for $\varphi^*$ being the result of applying the map combinator to the projection function for $\varphi$. The safe pairing property is proved by using the induction assumption, and using the function 
 \begin{align*}
 (\Sigma_1 \times \cdots \times \Sigma_n)^* \to \Sigma_1^* \times \cdots \times \Sigma_n^*,
 \end{align*}
 which can easily be seen to be derivable. 
 \item The case when $\Sigma$ is a co-product $\Sigma_1 + \Sigma_2$ is proved similarly to the list case. Here, we use a bijective correspondence between quantifier-free formulas $\varphi$ over the vocabulary of $\Sigma$ with pairs $(\varphi_1,\varphi_2)$, where $\varphi_i$ is a quantifier-free formula over the vocabulary of $\Sigma_i$. 
 \item The case when $\Sigma$ is a product $\Sigma_1 \times \Sigma_2$ is proved similarly to the co-product case. Again, there is a bijective correspondence between quantifier-free formulas $\varphi$ over the vocabulary of $\Sigma$ with pairs $(\varphi_1,\varphi_2)$, where $\varphi_i$ is a quantifier-free formula over the vocabulary of $\Sigma_i$. For the existence of such a bijective correspondence, we use the assumption that the type is in disjunctive normal form. Thanks to the assumption, the vocabulary has no nullary relations; if there would be nullary relations then there could be some communication between the two coordinates in the product. 
 \end{itemize}
\end{proof}

\paragraph*{Completeness.} 
Consider a quantifier-free interpretation 
\begin{align*}
f : \Sigma \to \Gamma.
\end{align*}
Let $\varphi$ be the universe formula of $f$, and let $\Sigma|\varphi$ be the type obtained by applying Lemma~\ref{lem:safe-pairing}. We write $\dom f$ for this type. The corresponding function in the decomposition as in item~\ref{it:safe-pairing-factors} is then 
\begin{align*}
f|\dom f : \dom f \to \Gamma.
\end{align*}
We will use the following terminology for this decomposition: the type $\Sigma|f$ will be called the \emph{reduced domain of $f$}, the projection will be called the \emph{domain reduction of $f$}, and the function $g$ will be called \emph{reduced $f$}. Here is a diagram that displays this terminology 
\[
 \begin{tikzcd}
 \Sigma 
 \ar[rr,"f"]
 \ar[dr,"\text{domain reduction of $f$}"']
 && \Gamma\\
 & \text{reduced domain of $f$}
 \ar[ur,"\text{reduced $f$}"']
 \end{tikzcd}
 \]
 
Because the domain reduction is derivable, and derivable functions are closed under composition, it is enough to show that for every quantifier-free interpretation, its reduced version is derivable. This will be shown in the following lemma.

\begin{lemma}\label{lem:qf-completeness-lemma} For every  quantifier-free interpretation
    \begin{align*}
    f : \Sigma \to \Gamma
    \end{align*}
    with universe formula $\varphi$, one can derive the function
    \begin{align*}
    f|\varphi : \Sigma|\varphi \to \Gamma
    \end{align*}
    from item~\ref{it:safe-pairing-factors} in Lemma~\ref{lem:safe-pairing}.
\end{lemma}
\begin{proof}
    The lemma is proved by structural induction on the input and output types. In the induction step, we will replace either the input or output type by a simpler one. 
 The induction step is shown in Sections~\ref{sec:some-output-co-product}--\ref{sec:input-type-list} below, which consider the following cases:
 \begin{description}
 \item[\ref{sec:input-type-co-product}]the input type is a co-product;
 \item[\ref{sec:some-output-co-product}] the output type is a co-product;
 \item[\ref{sec:some-output-product}] the output type is a product;
 \item[\ref{sec:input-type-one}] the input type is $0$ or $1$;
 \item[\ref{sec:input-type-list}] the input type is a list;
 \item[\ref{sec:input-type-product}] the input type is a product.
 \end{description}
 These cases are exhaustive, i.e.~at least one of them always applied, but they are not disjoint. When applying some case, we assume that none of the previous cases can be applied. The induction basis corresponds to case~\ref{sec:input-type-one}.
 
 \subsubsection{The input type is a co-product}
 \label{sec:input-type-co-product}
 In the representation of the co-product type from Definition~\ref{def:structures-for-string-types}, the information about whether the structure comes from the first or second case is stored in a nullary predicate. Therefore, by a straightforward syntactic manipulation of quantifier-free interpretations, from a quantifier-free interpetation
 \begin{align*}
 f : \Sigma_1 + \Sigma_2 \to \Gamma,
 \end{align*}
 we can obtain two quantifier-free interpretations
 \begin{align*}
 f_1 : \Sigma_1 \to \Gamma \qquad f_2 : \Sigma_2 \to \Gamma
 \end{align*}
 which describe the behaviour of $f$ on inputs from $\Sigma_1$ and $\Sigma_2$, respectively. Let $\varphi$ be the universe formula of $f$, and let $\varphi_1$ and $\varphi_2$ be the universe formulas of $f_1$ and $f_2$. By induction assumption, we can derive 
 \begin{align*}
 f_i|\varphi_i : \Sigma_i|\varphi_i \to \Gamma
 \end{align*}
 and derive their reduced versions. Since by definition we have 
 \begin{align*}
 (\Sigma_1 + \Sigma_2) | \varphi = \Sigma_1|\varphi_1 + \Sigma_2|\varphi_2,
 \end{align*}
 we can combine these two derivations into a derivation $f|\varphi$, by using the combinator 
 \begin{align*}
 \combinator {\Delta_1 \to \Gamma \quad \Delta_2 \to \Gamma}{\Delta_1 + \Delta_2 \to \Gamma}{cases,}
 \end{align*}
 which itself can be derived using functoriality of $+$ and the co-diagonal.

 \subsubsection{The output type is a co-product} 
 \label{sec:some-output-co-product}
 Consider a function 
 \begin{align*}
 f : \Sigma \to \Gamma_1 + \Gamma_2
 \end{align*}
 whose output type is a co-product. In this case, we assume that the previous case cannot be applied, i.e.~the input type is not a co-product.

 To produce the output structure, we need to define the nullary predicate that says which of the two cases in the output type is used. In a quantifier-free interpretation, this nullary predicate is defined by a quantifier-free formula, with no free variables, which is evaluated in the input structure. Since there are no nullary predicates in the input structure (because otherwise, the input type would be a co-product, and we could apply the case from the previous section), it follows that this quantifier-free formula is either ``true'' or ``false''. This means that the function $f$ must always use the same variant $\Gamma_1$ or $\Gamma_2$ in the co-product from the output type, regardless of the choice of input structure. Therefore, we can replace $f$ by a corresponding function of type $\Sigma \to \Gamma_i$, apply the induction assumption, and conclude by using composition and the co-projection.

 \subsubsection{The output type is a product}
 \label{sec:some-output-product}
 Consider a function 
 \begin{align*}
 f : \Sigma \to \Gamma_1 \times \Gamma_2
 \end{align*}
 whose output type is a product. We split this function into two quantifier-free interpretations 
 \begin{align*}
 f_1 : \Sigma \to \Gamma_1 \quad f_2 : \Sigma \to \Gamma_2,
 \end{align*}
 which produce the two coordinates in the output of $f$.
 These two functions must have disjoint universe formulas, since otherwise the same element in the output structure would belong to both coordinates of a pair. We can apply the induction assumption, and then combine these derivations into a derivation of  $f$ by using safe pairing from Lemma~\ref{lem:safe-pairing}.

 \subsubsection{The input type is $0$ or $1$}
 \label{sec:input-type-one}
 By cases~\ref{sec:some-output-co-product} and~\ref{sec:some-output-product}, we can assume that the output type of the unique function in the family is either $0$, $1$, or a list type $\Gamma^*$. 
 
 When the output type is $0$ or $1$, then we are dealing with a quantifier-free interpretation which has one of the types 
 \begin{align*}
 0 \to 0 \quad 0 \to 1 \quad 1 \to 0 \quad 1 \to 1.
 \end{align*}
 There is no quantifier-free interpretation of the type $1 \to 0$, and for the remaining types there is exactly one quantifier-free interpretation, which is easily seen to be derivable. 

 We are left with the case when the output type is $\Gamma^*$. If the input type is $0$, then the quantifier-free interpretation necessarily produces the empty list, and it is therefore derivable. If the input type is $1$, then the function always produces the same output, which is either the empty list, in which case it can be derived using the list constructor, or a singleton list $[A]$ for some fixed structure $A \in \Gamma$. In the singleton case, we can use the induction assumption to derive the function $1 \mapsto A$, and pack the result as a list using the list unit operation.

 \subsubsection{The input type is a list}
 \label{sec:input-type-list} 
 We now arrive at the most interesting case in the proof, which is when the input type is a list  $\Sigma^*$. 
 Because the previously studied cases~\ref{sec:some-output-co-product} and~\ref{sec:some-output-product} cannot be applied, the output type is one of $0$, $1$, or $\Gamma^*$.
 When the output type is $0$, there is only one possible function, which is easily derivable. The output type $1$ is impossible, since the function could not handle an empty list on the input. We are left with a list-to-list function. To prove the inductive step for such functions, we use the analysis from the following claim.

 \begin{claim}\label{claim:simplify-list-to-list}
 For every quantifier-free interpretation 
 \begin{align*}
 f : \Sigma^* \to \Gamma^*
 \end{align*}
 one can find quantifier-free interpretations 
 \begin{align*}
 f_1,\ldots,f_k : \Sigma^* \to \Gamma^*
 \end{align*}
with disjoint universe formulas  such that $f$ is equal to
 \begin{align*}
 A \in \Sigma^*
 \quad \mapsto \quad \myunderbrace{f_1(A) \cdots f_k(A)}{list concatenation}
 \end{align*}
 and each $f_i$ has one of the following properties:
 \begin{enumerate}
 \item \label{it:all-outputs-singleton} all output lists of $f_i$ have length at most one.
 \item \label{it:flat-map} there is some quantifier-free interpretation 
 \begin{align*}
 g : \Sigma \to \Gamma^*
 \end{align*}
 such that $f_i$ is equal to 
 \begin{align*}
 [A_1,\ldots,A_n] \mapsto \myunderbrace{g(A_1) \cdots g(A_n)}{list concatenation}
 \end{align*}

 \item \label{it:flat-map-reverse} as in item~\ref{it:flat-map}, but with reverse list order $g(A_n) \cdots g(A_1)$.
 \end{enumerate}
 \end{claim}

Before proving the claim, we use it to complete the induction step of the lemma in the present list-to-list case. Apply Claim~\ref{claim:simplify-list-to-list} to the function $f$, yielding a decomposition into functions $f_1,\ldots,f_k$. The induction assumption can be applied to these functions, since item~\ref{it:all-outputs-singleton} in the claim gives a smaller output type (namely $\Gamma$ instead of $\Gamma^*$ for the only list element), while the remaining two items give smaller input types. Finally, these derivations can be combined into a derivation of $f$, using the pairing operation from Lemma~\ref{lem:safe-pairing}, the function for list concatenation from Figure~\ref{fig:binary-concatenation-string-diagram}, and the prime function
\begin{align*}
 \primefunction {(\Sigma \times \Gamma)^*} \to {\Sigma^* \times \Gamma^*} {list distribute}
\end{align*} 
which is used to separate the domains of the functions $f_1,\ldots,f_k$ from the input list.
It remains to prove the claim.

\begin{proof}[of Claim~\ref{claim:simplify-list-to-list}]
 Consider the universe formula $\varphi(x)$ of $f$. Decompose this formula as a finite union 
 \begin{align*}
 \varphi(x) = \bigvee_{\sigma \in \Phi} \sigma(x)
 \end{align*}
 of quantifier-free theories as in Definition~\ref{def:rank-free-theory}, i.e.~quantifier-free formulas that specify all relations satisfied by $x$. 
 Take some input structure in $\Sigma^*$. For elements of this structure that satisfy the universe formula, there are two orders: the \emph{input order} that describes the order in the input list
 \begin{align*}
 A = [A_1,\ldots,A_n] \in \Sigma^*
 \end{align*}
 and the \emph{output order} that describes the order in the output list
 \begin{align*}
 f(A) = [B_1,\ldots,B_m] \in \Gamma^*.
 \end{align*} In the proof of the claim, we will analyze the relationship between these two orders. Both of these orders are reflexive, total, and transitive, but not necessarily anti-symmetric, since two elements may belong to the same list element.

 For an element $a$ in an input structure $A \in \Sigma^*$ that satisfies the universe formula $\varphi(x)$, the \emph{unary theory} of $a$ is defined to be the unique quantifier-free theory $\sigma \in \Phi$ that is satisfied by $a$. 
 If $a$ is strictly smaller than $b$ in the input order, then by compositionality, the output order on $a$ and $b$ will be uniquely determined by the unary theories of the two individual elements $a$ and $b$. This means  that exactly of the following three implications must hold
 
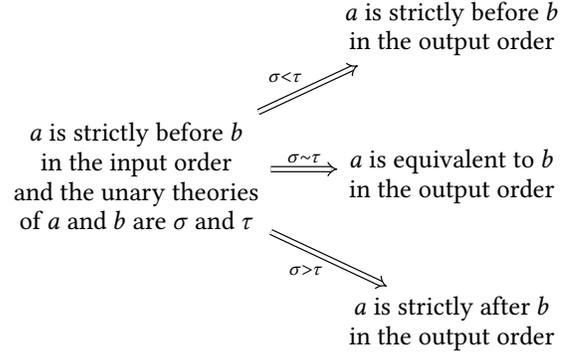
\[
 \begin{tikzcd}
    & \txt{$a$ is strictly before $b$\\ in the output order}\\
 \txt{$a$ is strictly before $b$\\ in the input order \\ and the unary theories \\ of $a$ and $b$ are $\sigma$ and $\tau$}
 \ar[ur,Rightarrow,"\sigma < \tau"]
 \ar[r,Rightarrow,"\sigma \sim \tau"]
 \ar[dr,Rightarrow,"\sigma > \tau"']
 & \txt{$a$ is equivalent to $b$\\ in the output order}\\
 & \txt{$a$ is strictly after $b$\\ in the output order}\\
 \end{tikzcd}
 \]
 Depending on which implication holds, we write one of 
 \begin{align*}
 \sigma < \tau \quad \sigma \sim \tau \quad \sigma > \tau.
 \end{align*}
 Before continuing, we make two cautionary remarks about the notation involving the relations $<$ and $>$ described above. The first cautionary remark is that the notation is not symmetric, since  $<$ and $>$ describe relations that are not necessarily converses of each other. This is because one of the conditions $\sigma < \tau$ or $\tau > \sigma$ could be true without the other one being true. The second cautionary remark is that $\sigma < \tau$ is not necessarily obtained from some partial order by looking at strictly growing pairs. For example, we could have both $\sigma < \tau$ and $\tau < \sigma$. 
 
 To prove the claim, we make five observations about the relations $<$, $>$ and $\sim$. In these observations, we use \emph{partial equivalence relations}; a partial equivalence relation is defined to be a binary relation that is symmetric and transitive but not necessarily reflexive. Equivalence classes of partial equivalence relations are defined in the expected way; the only difference is that some elements of the domain might not belong to any equivalence class.

 \begin{enumerate}
 \item \label{it:sim-equivalence} The first observation is that $\sigma \sim \tau$ is a partial equivalence relation. It is easy to see that the relation $\sigma \sim \tau$ is transitive. We now argue that it is symmetric. (This is not immediately obvious.) Suppose that $\sigma \sim \tau$. Consider a list in $A \in \Sigma^*$ 
 with four distinguished elements 
 \begin{align*}
 \myunderbrace{a_1}{unary \\ type $\sigma$}
 \quad < \quad
 \myunderbrace{a_2}{unary \\ type $\tau$}
 \quad < \quad
 \myunderbrace{a_3}{unary \\ type $\sigma$}
 \quad < \quad
 \myunderbrace{a_4}{unary \\ type $\tau$}
 \end{align*}
 with the order relationship describing the input order. From the assumption on $\sigma \sim \tau$ we can conclude that three pairs (depicted by lines in the following diagram) belong to the same elements in the output list:
 
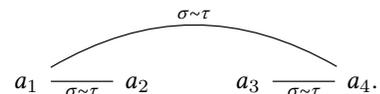
\[\begin{tikzcd}
 {a_1} 
 \ar[rrr, "\sigma \sim \tau", no head, bend left]
 \ar[r, "\sigma \sim \tau"', no head]
 & {a_2} & {a_3}
 \ar[r, "\sigma \sim \tau"', no head] & {a_4}.
 \end{tikzcd}\]
 Since belonging to the the same element in the output list is a transitive relation, we can deduce that $a_2$ and $a_3$ belong to the same element in the output list, thus establishing $\tau \sim \sigma$. 
 \item \label{it:ord-equivalence} The next observation is that $(\sigma < \tau \land \tau < \sigma)$ is a partial equivalence relation. It is symmetric by definition, and it is transitive because each of the two conjuncts is transitive.
 \item \label{it:ord-equivalence-other} By the same proof as in the previous item, $(\sigma > \tau \land \tau > \sigma)$ is a partial equivalence relation.
 
 \item We now show that the equivalence classes of the partial equivalence relations described in the first three observations are disjoint, and give a partition of 
 \begin{align*}
 \Phi = \Phi_1 \cup \cdots \cup \Phi
 \end{align*}
 of all unary types in $\Phi$. For every $\sigma \in \Phi$, we have exactly one of the cases $\sigma \sim \sigma$, $\sigma < \sigma$, or $\sigma > \sigma$. This proves that every $\sigma$ belongs to exactly one of the equivalence classes in the previous three items. 
 \item \label{it:order-equivalence-classes} The last observation is that the order on  equivalence classes in the previous item can be chosen so that for all $i < j$ we have 
 \begin{align*}
 \sigma \in \Phi_i \text{ and } \tau \in \Phi_j \quad \Rightarrow \quad \sigma < \tau.
 \end{align*} Let $\Phi_i$ and $\Phi_j$ be different equivalence classes from the previous item. For every $\sigma \in \Phi_i$ and $\tau \in \Phi_j$ we  have exactly one of the three cases 
 \begin{align*}
 \sigma < \tau \quad \text{or} \quad \sigma > \tau \quad \text{or} \quad \sigma \sim \tau.
 \end{align*}
 The third case cannot hold, since otherwise $\Phi_i$ and $\Phi_j$ would be in the same equivalence class from the first observation. Therefore, one of the two first cases must hold.  A short analysis, which is left to the reader, also shows that which of the two cases holds (first or second) does not depend on the choice of the $\sigma$ and $\tau$. This means that there is an unambiguous order relationship between $\Phi_i$ and $\Phi_j$, and this relationship can be used to prove item~\ref{it:order-equivalence-classes} of the claim.
 \end{enumerate}
 
 Let $\Phi_1,\ldots,\Phi_m$ be as in the last of the above observations. We know that for every input structure $A \in \Sigma^*$, the output list 
 can be decomposed as 
 \begin{align*}
 f(A) = f_1(A) \cdots f_n(A)
 \end{align*}
 where $f_i$ is the function obtained from $f$ by restricting the output elements to those that have type from $\Phi_i$ in the input structure. To complete the proof of the claim, we will show that each function $f_i$ has one of the three kinds in the statement of the claim.

 Suppose first that $\Phi_i$ is an equivalence class defined by $\sigma \sim \tau$ as in the first observation. This means that all outputs produced by $f_i$ are equivalent in the output order. Hence this $f_i$ is of kind~\ref{it:all-outputs-singleton} as in the statement of the claim.

 Suppose now that $\Phi_i$ is an equivalence class defined by $(\sigma < \tau \land \tau < \sigma)$ as in the second observation. This means that for every input list $A \in \Sigma^*$, if we take two elements $a$ and $b$ that have unary theory in $\Phi_i$, then 
 \[
 \begin{tikzcd}
 \text{$a$ is strictly before $b$ in the input order}
 \ar[d,Rightarrow]\\
 \text{$a$ is strictly before $b$ in the output order}
 \end{tikzcd}
 \]
 Hence this $f_i$ is of kind~\ref{it:flat-map} as in the statement of the claim. 
 
 A symmetric argument works for an equivalence class defined by $(\sigma > \tau \land \tau > \sigma)$, except that this time the output order is reversed, giving a function as in item~\ref{it:flat-map-reverse} of the lemma.
\end{proof}
 \subsubsection{The input type is a product}
 \label{sec:input-type-product}
 The final case in the proof of Lemma~\ref{lem:qf-completeness-lemma} is when the input type is a product. Since all types are in disjunctive normal form, the input type is a product
 \begin{align*}
 \Sigma = \Sigma_1 \times \cdots \times \Sigma_m
 \end{align*}
 where each $\Sigma_i$ is either $1$ or a list. (The type $0$ can be removed from a product.) Because the previously studied cases~\ref{sec:some-output-co-product} and~\ref{sec:some-output-product} about output types that are products or co-products cannot be applied, the output type is either $0$, $1$, or a list type $\Gamma^*$. 

 If the output type is $0$, then the function is easily derivable. 

 Consider now the case when the output type is $1$. It cannot be the case that each of the input types $\Sigma_1,\ldots,\Sigma_m$ is a list, since the quantifier-free interpretation would be unable to handle the case when all lists are empty. Therefore, one of the input types is the unit type $1$, and the conclusion of the lemma can be proved by using $1 \to 1$.

 We are left with the case when the ouput type is of the form $\Gamma^*$. Here, we proceed in the same way as in Section~\ref{sec:input-type-list}, with the corresponding version of Claim~\ref{claim:simplify-list-to-list} being the following claim. The proof of the claim, which uses a similar analysis of unary quantifier-free theories as in Claim~\ref{claim:simplify-list-to-list}, is left to the reader. 

 \begin{claim}
 For every quantifier-free interpretation 
 \begin{align*}
 f : \myunderbrace{\Sigma_1 \times \cdots \times \Sigma_m}{$\Sigma$} \to \Gamma^*
 \end{align*}
 one can find quantifier-free interpretations 
 \begin{align*}
 f_1,\ldots,f_k : \Sigma \to \Gamma^*
 \end{align*}
 with disjoint universe formulas such that $f$ is equal to
 \begin{align*}
 A \in \Sigma
 \quad \mapsto \quad \myunderbrace{f_1(A) \cdots f_k(A)}{list concatenation}
 \end{align*}
 and each $f_i$ has one of the following properties:
 \begin{enumerate}
 \item all output lists of $f_i$ have length at most one; or
 \item $f_i$ factors through the projection
 \begin{align*}
 \Sigma_1 \times \cdots \times \Sigma_m \to \Sigma_j \qquad \text{for some }j \in \set{1,\ldots,m}.
 \end{align*}
 \end{enumerate}
 \end{claim}

This completes the last of the cases in the induction step, and thus also the proof of the lemma, which also completes the proof of Theorem~\ref{thm:qf-system}.
\end{proof}

\section{Completeness for polyregular functions}
\label{sec:completeness}

In this section, we prove the completeness of the system in Theorem~\ref{thm:main}, i.e.~we show that every polyregular function can be weakly derived. This implication is the less interesting one, since our system is designed to be powerful, i.e.~it should be easy to derive functions in it. We will deduce the completeness of our system with fold from another completeness result that uses a system without fold.

We begin by describing the system that we reduce to. It has all of the combinators from Figure~\ref{fig:combinator-quantifier-free}, and its prime functions are contained in those from Figure~\ref{fig:prime-quantifier-free} plus certain additional functions that are described in Figure~\ref{fig:additional-polyregular-primes}. The first three primes from Figure~\ref{fig:additional-polyregular-primes} have already been discussed in the paper, so we only explain the block and split functions. The split function of type 
\begin{align*}
\Sigma^* \to (\Sigma^* \times \Sigma^*)^*
\end{align*}
outputs all possible ways of splitting the input list into (prefix, suffix) pairs, as explained in the following example:
\begin{gather*}
[1,2,3] \\ \downmapsto \\ [([],[1,2,3]),([1],[2,3]),([1,2],[3]),([1,2,3],[])].
\end{gather*}
The other additional function is the block function of type 
\begin{align*}
(\Sigma+\Gamma)^* \to (\Sigma^* + \Gamma^*)^*,
\end{align*}
which blocks the elements of the input list into maximal blocks of same type, as illustrated in the following example that uses numbers for elements of $\Sigma$ and letters for elements of $\Gamma$:
\begin{gather*}
 [1,2,a,3,4,5,b,c] \\ \downmapsto \\ [[1,2],[a],[3,4,5],[b,c]].
 \end{gather*}

\begin{figure}
 \begin{align*}
 \primefunction{G^*}{\to}{G}{group multiplication}\\
 \primefunction{\Sigma}{\to}{\Sigma \times \Sigma}{diagonal}\\
 \primefunction{\Sigma^*}{\to}{1 + \Sigma \times \Sigma^*}{list destructor}\\
 \primefunction{(\Sigma+\Gamma)^*}{\to}{(\Sigma^* + \Gamma^*)^*}{block}\\
 \primefunction{\Sigma^*}{\to}{(\Sigma^* \times \Sigma^*)^*}{split}
 \end{align*}
 \caption{\label{fig:additional-polyregular-primes} Additional polyregular prime functions from~\cite{bojanczykPolyregularFunctions2018}.}
\end{figure}

\begin{theorem}\cite[p.~64]{bojanczykPolyregularFunctions2018}\label{thm:completeness-without-fold}
 A function between list types is polyregular if and only if it can be derived using the prime functions and combinators from the quantifier-free system Theorem~\ref{thm:qf-system}, plus the prime functions from Figure~\ref{fig:additional-polyregular-primes}.
\end{theorem}

In contrast to the system with fold from this paper, the system from the above theorem was designed to be minimal, and therefore, the completeness proof for the system with fold will be a simple corollary of completeness of the system from the above theorem. Thanks to Theorem~\ref{thm:completeness-without-fold}, to prove the completeness result for our system with fold, it is enough to show that (a) all prime functions in Theorem~\ref{thm:completeness-without-fold} are weakly derivable; and (b) the combinators in Theorem~\ref{thm:completeness-without-fold} preserve the weakly derivable functions. 

\paragraph*{Combinators.}
Consider first (b), about the combinators. The combinators are those from Figure~\ref{fig:combinator-quantifier-free}. There is one combinator for function composition, and three combinators for functoriality. The combinators for functoriality are dealt with using the prime functions about $!$ commuting with the remaining constructors. The combinator for function composition is explained in the following diagram: 
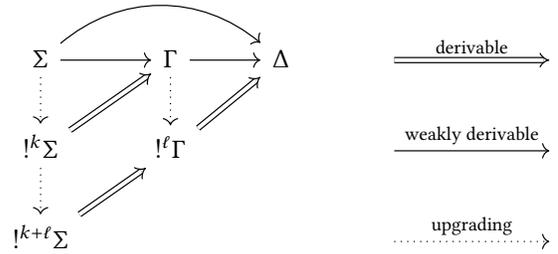
\[\begin{tikzcd}
	\Sigma & \Gamma & \Delta & \ &\ &\ \\
	{!^k \Sigma} & {!^\ell \Gamma}& \ &\ &\ &\ \\
	{!^{k+\ell} \Sigma} & \ &\ &\ & \ &\ 
	\arrow[from=1-1, to=1-2]
	\arrow[from=1-2, to=1-3]
	\arrow[dotted, from=1-1, to=2-1]
	\arrow[dotted, from=2-1, to=3-1]
	\arrow[Rightarrow, from=2-1, to=1-2]
	\arrow[dotted, from=1-2, to=2-2]
	\arrow[Rightarrow, from=2-2, to=1-3]
	\arrow[Rightarrow, from=3-1, to=2-2]
	\arrow[bend left = 40, from=1-1, to=1-3]
 \arrow[Rightarrow, "\text{derivable}", from=1-4, to=1-6]
 \arrow[ "\text{weakly derivable}", from=2-4, to=2-6]
 \arrow[dotted, "\text{upgrading}", from=3-4, to=3-6]
\end{tikzcd}\]

\paragraph*{Prime functions.}
Consider now (a), about the prime functions. Clearly all prime functions in the quantifier-free system are weakly derivable, since they are even strongly derivable. Weak derivability of the additional functions for group multiplication and the list destructor was already discussed in Examples~\ref{ex:groups-go-away} and~\ref{ex:deconstructor}. The diagonal function can easily be weakly derived using absorption. We are left with the split and block function. 

\begin{lemma}\label{lem:split-block}
 Split and block are weakly derivable.
\end{lemma}
\begin{proof}
	Apply two times in a row the weakly derivable prefixes function Claim~\ref{claim:reverse-prefix} to a list of the form
\begin{align*}
 [a_1,\ldots,a_n] \in \Sigma^*.
\end{align*} 
The output  is a list in $\Sigma^{***}$ of length $n$ whose $i$-th element is 
 \begin{align}\label{eq:split-intermediate}
[[a_1,\ldots,a_n],[a_1,\ldots,a_{n-1}], \ldots,  [a_1,\ldots,a_i]].
 \end{align}
 Since weakly derivable functions are closed under composition, this output can be produced by a weakly derivable function. Since weakly derivable functions are also closed under map, to complete the proof that split is weakly derivable, it remains to show that a weakly derivable function can transform the $i$-th element in~\eqref{eq:split-intermediate} into the corresponding element in the output of split, namely
 \begin{align}\label{eq:split-final}
 ([a_1,\ldots,a_i],[a_{i+1},\ldots,a_n]).
 \end{align}
 This is done as follows: we reverse the list in~\eqref{eq:split-intermediate}, and then apply the weakly derivable list destructor to get the pair consisting of the head and tail: 
 \begin{align*}
 \text{head} = & [a_1,\ldots,a_i]\\
 \text{tail} = & [[a_1,\ldots,a_{i+1}],[a_1,\ldots,a_{i+2}], \ldots,  [a_1,\ldots,a_{n}]].
 \end{align*}
 The head is already in the form required by~\eqref{eq:split-final}. 
 In the tail, we replace each list element (which itself is a list) by its last element; this can be done using map and the weakly derivable function that replaces a list with its last element. 

 We now turn to the block function. One approach is to derive the block function from split -- thus showing that it is not needed in the system. This is shown in~\cite[p.90]{bojanczykPolyregularFunctions2018}. However, since we will later use a system that uses block but not split, we show how to derive block directly. 
 To compute the block function, we use an automaton where the input alphabet is $\Sigma + \Gamma$, the state space is 
 \begin{align*}
 \Delta = (\Sigma^* + \Gamma^*)^* \times \myunderbrace{(\Sigma^* + \Gamma^*)}{most recent block}
 \end{align*}
 and the transition function is illustrated in the following diagram (by symmetry, we only draw the left half):
 \mypic{11}
 In the diagram, the unit function is the function $x \mapsto [x]$ which can be derived as in Figure~\ref{fig:binary-concatenation-string-diagram}.
 If we set the initial state of the above automaton to be a pair of empty lists (the second one having type, say, $\Sigma^*$), then after reading a list in $!(\Sigma + \Gamma)^*$, its state will store the output of the block operation, except that the last list element will be held separately and will need to be added using append. 
\end{proof}


\subsection{A smaller system}
A corollary of the completeness proof is Theorem~\ref{thm:reduced}, which shows that certain primes and combinators can be removed from the system in Theorem~\ref{thm:main}, while keeping it complete. We remove the map combinator, as well as all quantifier-free functions from Figure~\ref{fig:prime-quantifier-free} that involve the list type, namely the functions
\begin{align*}
 \primefunction {\Sigma^* \times \Sigma} \to {\Sigma^*} {append}\\
 \primefunction {\Sigma^{*}} \to {\Sigma^*} {reverse}\\
 \primefunction {\Sigma^{**}} \to {\Sigma^*} {concat}\\
 \primefunction {\Sigma} \to {\Sigma \times \Gamma^*} {create empty} \\
 \primefunction {(\Sigma \times \Gamma)^*} \to {\Sigma^* \times \Gamma^*} {list distribute}
\end{align*}
In their place, we have only two functions
\begin{align*}
 \primefunction {1 + \Sigma} \to {\Sigma^*} {lists of length at most one}\\
 \primefunction {\Sigma^* \times \Sigma^*} \to {\Sigma^*} {binary list concatenation.}
 \end{align*}

We will show that the smaller system remains complete, because it can weakly derive the removed functions, and furthermore, the weakly derivable functions in the smaller system are closed under the map combinator.

\begin{proof}[of Theorem~\ref{thm:reduced}]
    Consider first the prime functions that are removed from the smaller system.
The append function can be (strongly) derived in the smaller system. Using append, we can (strongly) derive the left list constructor, whose safe folding gives the list reversal in type 
\begin{align*}
!\Sigma^* \to \Sigma^*.
\end{align*}

is obtained by composing a co-projection with the right list constructor. Applying the safe fold combinator to the left list constructor (after swapping the order of its arguments) shows that the reverse function can be derived in type 
\begin{align*}
    !\Sigma^* \to \Sigma,
\end{align*} and hence it is weakly derivable. The concat function is derived in type 
\begin{align*}
!\Sigma^{**} \to \Sigma^*
\end{align*}
by folding binary list concatenation. To weakly derive the create empty function, we observe that for every type $\Sigma$ we can derive the unique function
\begin{align*}
!\Sigma \to 1,
\end{align*}
and this derivation can be used together with absorption to derive the create empty function in type 
\begin{align*}
!\Sigma \to \Sigma \times \Gamma^*.
\end{align*}
Finally, the list distribute function can be derived in type 
\begin{align*}
    !(\Sigma \times \Gamma)^* \to \Sigma^* \times \Gamma^*
\end{align*}
by a straightforward application of safe fold.

Finally, we can also eliminate the map combinator (functoriality of $*$), since using safe fold we obtain a version of  the map combinator in type 
\begin{align*}
 \frac{\Gamma \to \Sigma} { !\Gamma^* \to \Sigma^*} & & \text{weak map},
\end{align*}
which is strong enough to replace the usual map combinator in the completeness proof of the system in Theorem~\ref{thm:main}. Summing, up we can reduce the system as stated in the present Theorem~\ref{thm:reduced}, thus completing its proof.
\end{proof}

For easier reference, the system in the above theorem is described in Figure~\ref{fig:reduced-system}.

\begin{figure}
 \begin{align*}
    \primefunction {1 + \Sigma} \to {\Sigma^*} {lists of length at most one}\\
 \primefunction {\Sigma^* \times \Sigma^*} \to {\Sigma^*} {binary list concatenation}\\[1.5ex] 
 \primefunction{\ \upgradelong{\Gamma + \Sigma}}{\leftrightarrow}{\upgrade \Gamma + \upgrade \Sigma }{{! commutes with $+$}}\\
 \primefunction{\ \upgradelong{\Gamma \times \Sigma}}{\leftrightarrow}{\upgrade \Gamma \times \upgrade \Sigma }{{! commutes with $\times$}}\\
 \primefunction{\ (\upgrade \Gamma)^*}{\leftrightarrow}{\upgradelong{\Gamma^*}}{{! commutes with $*$}}\\
 \primefunction{\ \upgrade{\Gamma}}{\rightarrow}{\upgrade \Gamma \times \Gamma }{{absorption}} \\[1.5ex]
 \primeqf
 \end{align*}
 \begin{align*}
 \foldcombinator \\[1.5ex]
 \frac{\Gamma \to \Sigma \quad \Sigma \to \Delta} { \Gamma \to \Delta} & & \text{function composition}\\[1.5ex]
 \frac{\Gamma_1 \to \Sigma_1 \quad \Gamma_2 \to \Sigma_2}{\Gamma_1 \times \Gamma_2 \to \Sigma_1 \times \Sigma_2} & & \text{functoriality of $\times$} \\[1.5ex]
 \frac{\Gamma_1 \to \Sigma_1 \quad \Gamma_2 \to \Sigma_2}{\Gamma_1 + \Gamma_2 \to \Sigma_1 + \Sigma_2} 
 & & \text{functoriality of $+$} \\[1.5ex] 
 \qquad \frac{\Gamma \to \Sigma} { !\Gamma \to !\Sigma} & & \text{functoriality of !} \\[1.5ex]
 \end{align*}
\caption{\label{fig:reduced-system} A complete system for weakly deriving the polyregular functions. The safe fold combinator can only be applied when the type $\Gamma$ has grade $<k$.}
\end{figure}

\section{Soundness for polyregular functions}
In this section, we prove the soundness implication in Theorem~\ref{thm:main}. We prove that every strongly derivable function is a  graded \mso interpretations. The prime functions from Figure~\ref{fig:prime-quantifier-free} are quantifier-free, and therefore they are a special case of graded \mso interpretations. The extra prime functions from Theorem~\ref{thm:main}, namely absorption and those about $!$ commuting with the remaining type constructors, are easily seen to be graded \mso interpretations. The combinators for functoriality are also easily seen to preserve graded \mso interpretations. There are two interesting cases, namely the combinators for function composition and safe fold.

\subsection{Function composition}
We first show that the graded \mso interpretations are closed under composition, as long as the input and output types are graded list types. In the proof, we use the following result about composition of (non-graded) \mso interpretations on (non-graded) list types. 

\begin{lemma}\label{lem:mso-theories-composition}
    Let $f: \Sigma \to \Gamma$ be a non-graded \mso interpretation between non-graded list types, with the underlying functor being $F$.  
    For every $k,r \in \set{0,1,\ldots}$, the following function is \mso definable. 
    \begin{description}
        \item[Input] A structure $A \in \Sigma$ with distinguished elements $\bar a \in A^k$.
        \item [Output] The rank $r$ \mso theory of the tuple $F(\bar a)$ in $f(A)$. 
    \end{description}
\end{lemma}
\begin{proof}
    This lemma is a corollary of the closure under composition of (non-graded) \mso interpretations for (non-graded) list types~\cite[Corollary 8]{msoInterpretations}. The cited result is non-trivial, and  depends on the fact that the input and output types are list types.
\end{proof}

Using the above lemma for non-graded intepretations, we prove closure under composition of graded \mso interpretations over graded list types. 
Consider two graded \mso interpretations
\[
\begin{tikzcd}
\Sigma 
\ar[r,"f"]
&
\Gamma 
\ar[r,"g"]
&
\Delta.
\end{tikzcd}
\]
where all types involved are graded. 
We want to show that their composition 
\begin{align*}
f;g : \Sigma \to \Delta
\end{align*}
is a graded \mso interpretation.  Let the corresponding polynomial functors be $F$ and $G$. Naturally, the polynomial functor for the composition is going to be the composition functor $F;G$. The partition into grades of the output elements will be inherited from the second functor $G$. 

It remains to prove that the composition $f;g$ satisfies the continuity condition from Definition~\ref{def:graded-interpretation}. The continuity condition says that  every $k,\ell\in \set{0,1,\ldots}$, and  for every input structure $A \in \Sigma$  with distinguished elements  $\bar a \in A^k$, 
\begin{enumerate}
    \item[(1)] the quantifier-free theory of $(F;G)(\bar a)$ in $(f;g)(A)|\ell$
\end{enumerate}
   is uniquely determined by the quantifier-free theory of $\bar a$ in $A|\ell$, and the \mso theory of $\bar a$ in $A|\ell+1$ for some suitable quantifier rank that depends only on $f$ and the parameters $k,\ell$. 

By the continuity condition for the second graded \mso interpretation $g$, we know that the quantifier-free theory in (1) is uniquely determined by 
\begin{enumerate}
    \item[(2)] the quantifier-free theory of $F(\bar a)$  in $f(A)|\ell$; and
    \item[(3)] the rank $r$ \mso theory of $F(\bar a)$ in $f(A)|\ell+1$.
\end{enumerate}
By the continuity condition for the first graded \mso interpretation $f$, we know that the quantifier-free theory in (2) is uniquely determined by 
\begin{enumerate}
    \item[(4)] the quantifier-free theory of $\bar a$  in $A|\ell$;
    \item[(5)] the rank $s$ \mso theory of $\bar a$ in $A|\ell+1$.
\end{enumerate}
for some quantifier rank $s$. Consider now the \mso theory in (3). We want to show that this theory is also determined by suitable quantifier-free and \mso theories in the original structure $A$.  Consider the (non-graded) \mso interpretation 
\begin{align*}
A| \ell+1 \qquad \mapsto \qquad f(A)|\ell+1,
\end{align*}
which is well-defined by the continuity condition for $f$. By applying  Lemma~\ref{lem:mso-theories-composition} to this interpretation, we see that the \mso theory in (3) is uniquely determined by the 
\begin{enumerate}
    \item[(6)] the rank $t$ \mso theory of $\bar a$ in $A|\ell+1$.
\end{enumerate}
for some quantifier rank $t$. Summing up, the quantifier-free theory in (1) is uniquely determined by the quantifier-free theory in item (4), and the \mso theories in items (5) and (6).  The latter two \mso theories are determined by the single \mso theory for the higher quantifier rank among $s$ and $t$. Summing up, we have proved the continuity condition for the composed graded \mso interpretation $f;g$.

\subsection{Safe fold}

We are left with showing that graded \mso interpretations are closed under the safe fold combinator.  All of the conceptual pieces are already in place, and we will simply show that the proof of Theorem~\ref{thm:fold-quantifier-free} works, with minor adjustments to take into account the added generality of graded structures.

Suppose $\Gamma$ is a type where all grades are $<k$, and  we apply the safe fold combinator to graded \mso interpretations of types
\begin{align*}
 !^{k}1 \to \Gamma \qquad \text{and} \qquad   \Gamma \times \Sigma \to \Gamma,
\end{align*}
yielding a function  of type 
\begin{align*}
 !^{k}\Sigma \to \Gamma.
\end{align*}
By choice of $k$, in the resulting function every element in the input structure has strictly bigger grade than every element in the ouput structure. For such functions, the continuity condition in Definition~\ref{def:graded-interpretation} becomes trivial, and there is no difference between graded and un-graded \mso interpretations. Therefore, in order to prove the soundess of fold, it is enough to show the following lemma, that applying fold to a graded \mso interpretation yields an (ungraded) \mso interpretation.

\begin{lemma}\label{lem:safe-fold}
    For every graded \mso interpretation 
\begin{align*}
\delta : \Gamma \times \Sigma \to \Gamma,
\end{align*} 
between graded list types, and every $B_0 \in \Gamma$,  the following function is an (ungraded) \mso interpretation 
\begin{align*}
A=\myunderbrace{[A_1,\ldots,A_n]}{list of structures in $\Sigma$, \\ with the grades forgotten} 
\qquad \mapsto \qquad  \myunderbrace{B_n}{defined based on $A$ \\ as in the proof of Claim~\ref{claim:qf-folder} }.
\end{align*}
    \end{lemma}

\begin{proof}
    Let $m$ be the maximal grade that appears in $\Gamma$, and let  the polynomial functor in the transition function $\delta$ be 
\begin{align*}
    F(A) = F_0(A) + \cdots + F_{m}(A) + A.
    \end{align*}
By the continuity condition for the graded \mso interpretation $\delta$,  the  elements of grade $\ell$ in   $B_n$ are the disjoint union of two sets:
\begin{enumerate}
    \item grade  $\ell$ elements in  $B_{n-1}$ or $A_n$;  or
    \item $F_\ell$ applied to  grade $> \ell$ elements in $B_{n-1}$ or $A_n$.
\end{enumerate}
By unfolding the inductive definition of $B_{n-1}$ in the first item of the above description, we see that the elements of grade $\ell$ in $B_n$ are the disjoint union of two sets:
\begin{enumerate}
    \item[1*.] grade $\ell$ elements in $B_0$ or $A_1,\ldots,A_n$; or 
    \item[2*.] $F_\ell$ applied to  grade $> \ell$ elements in $B_{i-1}$ or $A_i$ for some $i \in \set{1,\ldots,n}$.
\end{enumerate}
We will represent the elements that satisfy 1* or 2* as a subset of $G_\ell(A)$ for some  polynomial functor $G_\ell$. This functor is defined as follows by induction on $\ell$, in reverse order $m,\ldots,0$. Suppose that we want to define $G_\ell$ and assume that we have already defined $G_{\ell'}$ for $\ell' > \ell$. (In the induction basis of $\ell=m$ the assumption is empty.)
    To represent the elements in item 1*, we use the functor 
    \begin{align*}
        A + \myunderbrace{1 + \cdots + 1}{number of elements\\ in $B_0$ that have grade $\ell$}.
    \end{align*}
A tempting idea for item 2* is to use the functor
    \begin{align*}
    H_\ell(A) =  F_\ell( G_{\ell+1}(A) + \cdots + G_{m}(A) + \myunderbrace{A}{represents elements of grade $>m$ \\ in the input structure}).
    \end{align*}
    Unfortunately, this idea is not correct. The reason is that in item 2*, there is a dijsoint union ranging over $i \in \set{1,\ldots,n}$, and the disjointness of this union is not taken into account by $H_\ell$. The problem is that the universe of the structures $B_0,\ldots,B_n$ are not disjoint, and the functor $H_\ell$ can incorrectly identify elements that are obtained by applying $F_\ell$ to the same elements that appear in both $B_{i}$ and $B_{j}$ for $i \neq j$. To eliminate this problem, we will add an explicit identifier for the index $j$ to the functor. To view the index $i$ as an element of the input structure $A_i$, we use the first element in the universe of the corresponding list element $A_i$. Here, when we refer to the first element in the universe, we mean the natural linear order on the universe in a structure from a graded list type, which arises from the ordered nature of lists and pairs. Therefore, instead of $H_\ell(A)$, to represent item 2* we use the product $A \times H_\ell(A)$, with the  $A$ part representing the index $i$. 
    Summing up, the functor $G_\ell$ that describes elements in each $B_i$ is  
    \begin{align*}
G_\ell(A) =   A + \myunderbrace{1 + \cdots + 1}{number of elements\\ in $B_0$ that have grade $\ell$} +  A \times H_\ell(A).
       \end{align*}

    In the rest of this proof, we will view the universe of $B_n$ as being a subset of 
    \begin{align*}
    G(A) = G_0(A) + \cdots + G_m(A),
    \end{align*}
    with $G_\ell(A)$ representing the elements of grade $\ell$. The polynomial functor $G(A)$ will be the  polynomial functor for the \mso interpretation in the conclusion of the lemma. To conclude the proof of the lemma, we need to show that in \mso we can define which elements of $G(A)$ belong to the universe of $B_n$, and what relations from the output vocabulary are satisfied by tuples of such elements. In other words, we need to define in \mso the quantifier-free theory of tuples from $G(A)$ in the output structure. This is done in the following claim, which completes the proof of the lemma.

\begin{claim}\label{claim:continuity}
    For every $\ell,k \in \set{0,1,\ldots}$ the following 
    function is \mso definable:
    \begin{itemize}
    \item {\bf Input.} A structure $A \in \Sigma^*$ with  elements $\bar a \in A^k$.
    \item {\bf Output.} The quantifier-free theory of $G(\bar a)$ in $B_n|\ell$.
    \end{itemize}
    Furthermore, the output depends only on $A$ and $\bar a$ restricted to elements of grade at least $\ell$.
    \end{claim}


\begin{proof}
    Fix some $\ell$ and $k$ as in the statement of the claim. The claim is proved by induction on $\ell$, in reverse order $m,\ldots,0$. Suppose that we want to prove the claim for some grade $\ell$, and  assume that it has already been proved for strictly bigger grades.

    We use the same idea as in the proof of Claim~\ref{claim:qf-folder}.
    Consider a finite automaton, in which the states are all possible theories that arise by taking some $k$-tuple $\bar a$, and returning the quantifier-free theory of $G(\bar a)$ in some structure from $\Gamma$. This set of states is finite, since the length of the tuple and the vocabulary are fixed. 
    
    We will design an automaton with this set of states, together with an input string (which will be called the \emph{advice string}), so that it satisfies the following invariant: after reading the first $i$ letters of the advice string, the state of the automaton is the quantifier-free theory of $G(\bar a)$ in $B_i|\ell$. 
    
    The initial state of the automaton is determined by the invariant, it must be the quantifier-free theory of $G(\bar a)$ in $B_0$. Since the universe of $B_0$ is equal to $G(\emptyset)$, it follows that the initial state does not depend on the tuple $\bar a$ or the input structure $A$.

    We now describe the transition function of the automaton, as well as the advice string.
    By unfolding the definition of the graded \mso interpretation $\delta$, there is some quantifier rank $s$ such that the state of the automaton after reading $i$ letters is uniquely determined by the following four pieces of information:
    \begin{enumerate}
        \item the quantifier-free theory of $G(\bar a)$ in $B_{i-1}$,       \item the quantifier-free theory of $G(\bar a)$ in $A_{i}$,
        \item the rank $s$ \mso theory of  $G(\bar a)$ in $B_{i-1}|\ell+1$,        \item the rank $s$ \mso theory of  $G(\bar a)$ in $A_{i}|\ell+1$.\end{enumerate}
The first piece of information is the previous state of the automaton. The remaining  infomration will be the stored in the advice string; i.e.~the $i$-th letter of the advice string will contain the information described the last three items above. Note that the advice string can be computed in \mso, by the induction assumption. Therefore, since the automaton can be simulated in \mso, it follows that the last state of this automaton can be defined in \mso, thus proving the claim.
\end{proof}
    \end{proof}

\section{Proof of Theorem~\ref{thm:linear}}
In this section, we prove that the system in Theorem~\ref{thm:main} is sound and complete with respect to linear regular functions. 

\paragraph*{Soundness.}
    The soundness proof follows the same lines as the soundness proof in Theorem~\ref{thm:main}. The general idea is that we use graded \mso interpretations where all components have dimension at most one. This, however, on its own is not going to be enough. To see why, let us compare the two absorption functions
    \begin{align*}
        \myunderbrace{!\Sigma \to \Sigma \times !\Sigma}{not allowed}
        \quad 
        \myunderbrace{!\Sigma \to \Sigma \times \Sigma}{allowed}.
    \end{align*}
    Both of them have linear size increase -- each element of the input structure contributes two copies to the output structure. What is wrong with the function that is not allowed? The problem is that one of the copies has the same grade, and the other has lower grade. In the presence of folding, we can get an unbounded number of copies, by spawning a new lower grade copy in each iteration. This phenomenon will not occur in the allowed function, since both copies have lower grade.  The phenomenon discussed above is formalised in the following definition: 
    
    \begin{definition}\label{def:linear-graded-mso}
        A \emph{linear graded \mso interpretation} is a graded \mso interpretation in  which the underlying functor is linear, i.e.~all components have dimension one, and which furthermore satisfies the following \emph{downgrading} condition:  if an element of the  input structure has at least two copies in the output structure, then all of the copies have strictly lower grade.
    \end{definition}
    
    In the definition above, the copies of an element in the output structure are defined in the natural way; this definition makes sense when the functor is linear. For example, if the functor is 
    \begin{align*}
 A + A + A + 1 + 1
    \end{align*} 
    then each input element spawns at most three copies. The components of dimension zero, of which there are two in the above example, are not counted as copies of any input elment.
    
    To prove completeness of the system from Theorem~\ref{thm:linear}, we show that all functions that are strongly derived in it are linear graded \mso interpretations. The proof is a simple inducton on the derivation. The most interesting cases are composition and folding. For composition, we simply observe that the condition on lower grades from Definition~\ref{def:linear-graded-mso} is preserved under composition. 
    
    We are left with folding. where we use the following lemma, which is the same as Lemma~\ref{lem:safe-fold} except that the functions in the assumption and conclusion are required to be linear. In the assumption, we use linearity as defined in Definition~\ref{def:linear-graded-mso}, in particular the downgrading condition is assumed; in the conclusion we have an ungraded function, and therefore only the linearity of the functor and not the downgrading condition are assumed.

    \begin{lemma}\label{lem:linear-safe-fold}
        For every linear graded \mso interpretation 
    \begin{align*}
    \delta : \Gamma \times \Sigma \to \Gamma,
    \end{align*} 
    between graded list types, and every $B_0 \in \Gamma$,  the following function is an (ungraded) linear \mso interpretation 
    \begin{align*}
    A=\myunderbrace{[A_1,\ldots,A_n]}{list of structures in $\Sigma$, \\ with the grades forgotten} 
    \qquad \mapsto \qquad  \myunderbrace{B_n}{defined based on $A$ \\ as in the proof of Calim~\ref{claim:qf-folder} }.
    \end{align*}
        \end{lemma}
\begin{proof}
    We use the  same proof as in Lemma~\ref{lem:safe-fold}. However, there is one difficulty, which is that the functor $G$ defined in that proof is not linear, even if $\delta$ is linear. This is because of the product $A \times H_\ell(A)$ which is used to code indexes. In fact, the functor $G$ can have  arbitrarily high dimension. However, thanks to the downgrading condition on $\delta$, one show by induction that for every grade $\ell$ there is some constant $c_\ell \in \set{0,1,\ldots}$ such that for every grade $\ell$ element $a$ in the input structure, there are at most $c_\ell$ elements in the output structure which use $a$. Here, we say that an element uses $a$ if it belongs to $G(A)$ but not to $G(A \setminus \set a)$. Using this property, we can turn $G$ into a linear functor.
\end{proof}

    This finishes the soundness proof. Below, we give two completeness proofs.
    \medskip

    \begin{proof}[First completeness proof.]
        This proof uses  the \sst model from Example~\ref{ex:sst}, which is  complete for linear regular functions, in the case where the input and output types are strings over finite alphabets~\cite[Theorem 3]{alurExpressivenessStreamingString2010}. In Example~\ref{ex:sst}, we show how to weakly derive every \sst that uses each input letter at most once. To get the general form of \sst, where an input letter can be used a constant number of times, it is enough to generalize the model from Example~\ref{ex:sst} so that the initial function is weakly derivable, and the transition function can be derived in type 
        \begin{align*}
        \Delta \times !^k \Sigma \to \Delta
        \end{align*}
        for some $k$. With these relaxations, we get all copyless \sst, and retain weak derivability. This proof works only for functions of string-to-string type (admittedly, this is the case that we really care about), and for this reason we also present a second proof, which can also handle types such as strings of strings or pairs of strings.
    \end{proof}
    
\begin{proof}[Second completeness proof.] In this proof, similarly to the completeness proof from Theorem~\ref{thm:main}, we reduce to a known complete system. In the case of linear \mso interpretations, the corresponding known system is from~\cite{bojanczykRegularFirstOrderList2018}. 
    It is the same as in Theorem~\ref{thm:completeness-without-fold}, except that the split function is removed. In the completeness proof of Theorem~\ref{thm:main}, only the proof for split used general absorption (as opposed to linear absorption). Therefore, the system with linear absorption is complete for the linear regular functions.
\end{proof}

\end{document}